\documentclass[journal]{IEEEtran}
\ifCLASSINFOpdf
   \usepackage[pdftex]{graphicx}
\usepackage{tikz}
\else
\fi
\usepackage{amsmath}
\usepackage{amssymb,amsthm}

\newtheorem{proposition}{Proposition}
\newtheorem{lemma}{Lemma}
\newtheorem{corollary}{Corollary}
\newtheorem{conjecture}{Conjecture}

\begin{document}
%
\title{On the Age of Information of Processor Sharing Systems}
%
%
%

\author{Beñat~Gandarias, Josu~Doncel, Mohamad~Assaad
\thanks{B. Gandarias and J. Doncel are with University of the Basque Country, UPV/EHU.}
\thanks{M. Assaad is with Laboratoire des Signaux et Systèmes (L2S), CentraleSupelec, University of
Paris-Saclay.}
\thanks{This work has been partially supported by the Department of Education of the Basque Government, Spain through
the Consolidated Research Group MATHMODE (IT1294-19) and by the Marie Sklodowska-Curie grant agreement N.
777778.}
}

%
%

\markboth{On the Age of Information of Processor Sharing Queues}%
{Gandarias \MakeLowercase{\textit{et al.}}: On the Age of Information of Processor Sharing Queues}
%



\maketitle

\begin{abstract}
In this paper, we examine the Age of Information (AoI) of a source sending status updates to a monitor through a queue operating under the Processor Sharing (PS) discipline. While AoI has been well studied for various queuing models and policies, less attention has been given so far to the PS discipline.  We first consider the M/M/1/2 queue with and without preemption and provide closed-form expressions for the average AoI in this case. We overcome the challenges of deriving the AoI expression by employing the Stochastic Hybrid Systems (SHS) tool. We then extend the analysis to the M/M/1 queue with one and two sources and provide numerical results for these cases. Our results show that PS can outperform the M/M/1/1* queue in some cases.  
\end{abstract}


%
\IEEEpeerreviewmaketitle

\section{Introduction}

With the development of Internet of Things (IoT), there is an increasing interest nowadays in real-time monitoring, where a remote monitor is tracking the status of a source/sensor.  Age of Information (AoI) has been introduced  in \cite{kaul2012real} to capture the freshness of information in such contexts, networked control systems. This metric is defined as the time elapsed since the generation of the last correctly received packet at the monitor. Since its introduction, this metric has received particular interest from researchers and has been studied in various network models and scenarios. Since the evolution of the AoI over time exhibits a sawtooth pattern, researchers have focused on  AoI-dependent metrics computation, such as Average AoI, Peak AoI, etc. In particular, the average AoI has received a lot of attention and has been evaluated in various  continuous and discrete time network  models. In \cite{kaul2012real,kaul2012realCISS,yates2012realISIT,KamIT2016}, it was shown that the computation of the average AoI is a hard task in general settings since it needs the evaluation of the expected value of the product  of inter-arrival and response times, which are correlated random variables. The average AoI has been computed by considering specific source and response time models (although these models cover a wide range of scenarios), and the medium between the source and monitor has been modeled by a queuing system. For instance,  the authors in \cite{kaul2012real} derived the average AoI of the M/M/1 queue, M/D/1 queue, and D/M/1 queue models and obtained the best arrival rate of the packet update. The single source single destination M/M/1 queue under First Come First Served (FCFS) and Last Come First Served (LCFS) disciplines has been studied in \cite{kaul2012realCISS,yates2012realISIT}. The peak AoI has also been studied in various scenarios. For instance, average AoI and peak AoI have been analyzed for M/M/1/1 and M/M/1/2 queues \cite{costa2014ISIT,costa2016age}. To improve the AoI, a queue discipline called M/M/1/2*, in which a new arriving packet preempts a waiting one in the queue,  has been introduced in \cite{costa2014ISIT,costa2016age}. 

In addition to the aforementioned works that have focused on Poisson arrivals and exponential service time,  more general single queue models have been explored. In \cite{SoysalIT2021}, a G/G/1/1 is analyzed. The average aoI in a multi-source M/G/1/1 queue with packet management has been studied in \cite{ChenTcom2022,NajmAoI2018}, where it was shown that preempting packets does not reduce AoI for small service time variance coefficients.

The analysis of AoI was also studied in the case of multi-source single-server (e.g. M/M/1/2*) in \cite{yates2012real,yates2018age} and multi-source multi-server systems in \cite{bedewy2016optimizing,yates2018ISIT,JosuJCN2022}. 
While the aforementioned works have focused on predefined arrival and queuing disciplines, several studies have explored the optimal status update (information sampling and scheduling) policy in various scenarios, e.g. in single hop networks \cite{kadota2018scheduling,maatouk2020status,sun2018age},  multihop networks \cite{bedewy2017age,bedewy2019age}, etc. Interestingly, it has been shown in \cite{SunElifINFOCOM,SunElif2017}, for a single source and single destination scenario,   that zero-wait policy, where the source transmits a fresh update right when the previous one has been delivered, does not always minimize the AoI. For discrete-time multi-user networks, several Age-based scheduling solutions have also been developed, e.g. \cite{kadota2018scheduling,hsu2017age,maatouk2020optimality}. A whittle index based scheduling policy has been developed in \cite{kadota2018scheduling,hsu2017age}.  Such a Whittle index based policy has been proved to be asymptotically optimal in some cases \cite{SaadAoI2022,maatouk2020optimality}. Furthermore, there have been studies on energy-constrained updating, e.g. in energy harvesting context \cite{Arafa2018,Arafa2019,yates2015ISIT,Bacinoglu2018ISIT,Bacinoglu2015ITA}. It is worth mentioning that in addition to the above AoI metrics, there is an increasing interest recently in developing beyond-age metrics \cite{Chiariotti2022,zhong2018,maatouk2020age,escroc2019,kam2018towards,yates2021ISIT,SunCyr2019},  for example to capture the semantic of information \cite{ElifSemantic} such as value of information \cite{KostaVoI}, Age of Incorrect of Information (AoII)\cite{maatouk2020age,maatouk2020agenew,kriouile2021minimizing,ChenAoII2021}, Query Age of Incorrect Information (QAoII) \cite{Ayik2023}, etc. For more recent surveys of existing work the reader can refer to \cite{BookAoI2023,SunBook,Survey2021}. 

In this paper, we focus on the average AoI metric. We consider a single source single destination queuing model under Processor Sharing (PS) discipline. Under the Processor Sharing discipline, all the packets in the queue are served at the same speed, i.e., when there are $n$ packets in the queue, 
each packet gets a proportion of $1/n$ of the service capacity. Despite its extensive use and analysis since its introduction in \cite{kleinrock1975theory}, e.g. \cite{yashkov1987processor,kelly1979stochastic}, to the best of our knowledge, the PS queue has not been studied from the AoI perspective. In this paper, we provide analysis of the average AoI in a single queue system under the PS discipline. We make use of Stochastic Hybrid System (SHS) tool to overcome the challenges of analyzing the AoI under the PS discipline.   Specifically, the main contributions in this paper are as follows:
\begin{itemize}
	\item We first consider the M/M/1/2 queue without preemption and we provide an explicit expression of the AAoI for the PS discipline. 
	\item We then consider the M/M/1/2 queue with preemption, which we denote the M/M/1/2$^*$ queue. 
	We also provide an analytical expression of the AAoI when the queueing discipline is PS. 
	\item We show that, for the M/M/1/2 with and without preemption, the PS discipline outperforms the FGFS discipline in terms of AAoI. Moreover, for the M/M/1/2 queue, the AAoI for the FGFS discipline is, at most, 1.2 times worse than the AAoI for the PS discipline and, for the M/M/1/2$^*$ queue, it is at most 4/3 worse. 
	\item We prove that the AAoI of the M/M/1/2$^*$ queue is always smaller than the AAoI of the M/M/1/2 queue under PS, and, in fact, the AAoI of the M/M/1/2 queue is, at most, 5/3 worse than the AAoI of the M/M/1/2$^*$ queue. 
	\item We analyze  the AAoI of the M/M/1 queue for PS and FGFS (First Generated First Served) disciplines, and provide numerical results by solving the system of equations obtained by the SHS technique. 
	\item We analyze the case of two sources and provide numerical results, by solving the equations resulting from the SHS analysis. Interestingly, the results show that the PS discipline can outperform the M/M/1/1* in some cases.    
\end{itemize}
The remainder of the article is organized as follows. We present the system model in Section II. In Section III, we study the M/M/1/2 model without and with preemption, while the M/M/1 analysis is provided in Section IV. The multiple source scenario is given in Section V, and the conclusion is presented in Section VI. The proofs are given in the appendix.





\section{Model Description}
We consider a monitoring system in which there is a process of interest (i.e. the source) whose status needs to be observed timely by a remote monitor (i.e. the sink). To this end, packets containing information about the status of the system are generated at the source and are sent to the sink through a transmission channel. 
\begin{figure}[h!]
\centering
\includegraphics[width=0.8\columnwidth,clip=true,trim=20pt 550pt 20pt 170pt]{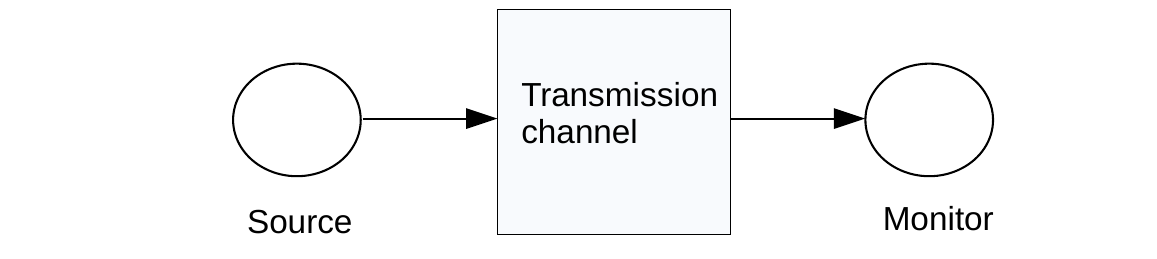}
\caption{A monitoring system example.}
\end{figure}

We assume that packet generation times at the source follow a Poisson process of rate $\lambda$ and that the transmission channel is a single server queue with exponential service times of rate $\mu$. 
The load of the system is $\rho=\lambda/\mu$.
Moreover, it is assumed that the transmission times from the source to the queue and from the queue to the monitor are both zero.

We consider that the queue serves the packet containing status updates of the system according to the Processor Sharing (PS) discipline. This means that all the packets in the queue are served at the same speed, i.e., when there are $n$ packets in the queue, 
each packet is served at rate $\mu/n$.

We consider the Age of Information as the metric of performance of the system. The Age of Information is defined as the time elapsed since the generation time of the last packet that has been delivered to the monitor successfully.  
More precisely, if $t_i$ is the generation time of the $i$-th packet and $L(t)$ is the index of the last successfully delivered packet
before time $t$, the Age of Information at time $t$ is defined as $\Delta(t)=t-t_{L(t)}$.

Our focus will be on the AAoI of divers queueing models, which we denote by $\Delta$ with a subindex that indicates the queueing model we refer to. For instance, when we study the AAoI of the M/M/1/2-PS queue, we denote it by $\Delta_{M/M/1/2-PS}$.

\section{The M/M/1/2 Queue With And Without Preemption}

\subsection{The M/M/1/2 Queue}
\label{sec:mm12}

We consider the M/M/1/2 queue without preemption. In this system, the maximum number of packets that can be stored in the queue is two. Besides, when a new packet arrives and there are two packets in the system, the incoming packet is discarded. Under the PS discipline, when there is a single packet in the queue, it is served at rate $\mu$, but when there
are two packets in the queue, both are served at rate $\mu/2$.

The following result characterizes the AAoI of the M/M/1/2 queue without preemption and under the PS discipline. The proof is available in Appendix~\ref{app:prop:mm12-ps}.

\begin{proposition}
The AAoI of the M/M/1/2-PS queue is
\begin{equation}
\Delta_{M/M/1/2-PS}=\frac{5\lambda^4+9\lambda^3\mu+8\lambda^2\mu^2+6\lambda\mu^3+2\mu^4}
{2 \lambda\mu(\lambda+\mu)(\lambda^2+\lambda\mu+\mu^2)}.
\label{eq:mm12-ps}
\end{equation}
\label{prop:mm12-ps}
\end{proposition}

In  \cite{costa2016age}, it is shown that the AAoI of the M/M/1/2 queue with the FGFS discipline and without preemption in waiting is
\begin{equation}
\Delta_{M/M/1/2-FGFS}= \frac{3\lambda^4+5\lambda^3\mu+4\lambda^2\mu^2+3\lambda\mu^3+\mu^4}
{\lambda\mu(\lambda+\mu)(\lambda^2+\mu^2+\lambda\mu)}.
\label{eq:mm12-fgfs}
\end{equation}

We now aim to analyze the benefits on the AAoI of the PS discipline with respect to the FGFS discipline by comparing \eqref{eq:mm12-ps} with \eqref{eq:mm12-fgfs}. The following result provides an analytical comparison of both expressions.
The proof is available in Appendix~\ref{app:prop:mm12-comparison}.

\begin{proposition}
We have that
\begin{equation}
1\leq \frac{\Delta_{M/M/1/2-FGFS}}{\Delta_{M/M/1/2-PS}}\leq 1.2.
\label{eq:mm12-comparison}
\end{equation}
\label{prop:mm12-comparison}
\end{proposition}

From this result, we conclude that the PS discipline outperforms the FGFS discipline and also that the AAoI when we consider the FGFS discipline is, at most, 1.2 times the AAoI of the PS discipline.

\subsection{The M/M/1/2$^*$ Queue}
\label{sec:mm12star}

We now consider the M/M/1/2 queue with preemption. In this system, when a new packet arrives to the system and there are two packets in the queue, the last update in the queue is replaced by the incoming one. Note that this is a big difference with respect to the M/M/1/2 queue without preemption that has been studied in the previous section. In fact, it is known that when we consider the Age of Information metric, preemption leads to a performance improvement with respect to a system without preemption \cite{bedewy2016optimizing}. In this section, we follow the  notation of \cite{costa2016age} and we denote by  M/M/1/2$^*$ the system under study in this section. 

Our goal is to extend the results of the previous section to the M/M/1/2$^*$ to analyze the impact of the preemption on the ratio of the AAoI of the FGFS discipline over the AAoI of the PS discipline. 

We now present the expression of the AAoI of the M/M/1/2$^*$. The proof of this result is postponed to Appendix~\ref{app:prop:mm12-star-ps}.

\begin{proposition}
The AAoI of the M/M/1/2$^*$-PS queue is
\begin{multline}
\Delta_{M/M/1/2^*-PS}= \\\frac{3\lambda^5+11\lambda^4\mu+15\lambda^3\mu^2+14\lambda^2\mu^3+8\lambda\mu^4+2\mu^5}
{2\lambda\mu(\lambda+\mu)^2(\lambda^2+\mu^2+\lambda\mu)}.
\label{eq:mm12-ps-star}
\end{multline}
\label{prop:mm12-star-ps}
\end{proposition}



We now aim to compare the AAoI of the M/M/1/2$^*$ queue under the PS discipline with the AAoI under the FGFS discipline. The expression of the former system has been shown in \cite{costa2016age}, and it is
\begin{multline}
\Delta_{M/M/1/2^*-FGFS}= \\\frac{2\lambda^5+7\lambda^4\mu+8\lambda^3\mu^2+7\lambda^2\mu^3+4\lambda\mu^4+\mu^5}
{\lambda\mu(\lambda+\mu)^2(\lambda^2+\mu^2+\lambda\mu)}.
\label{eq:mm12-fgfs-star}
\end{multline}

We focus on the ratio $\frac{\Delta_{M/M/1/2^*-FGFS}}{\Delta_{M/M/1/2^*-PS}}$. In the following result, we provide a lower-bound and an upper-bound of this ratio. The proof of this result is provided in Appendix~\ref{app:prop:mm12-star-comparison}.

\begin{proposition}
We have that
\begin{equation}
1\leq \frac{\Delta_{M/M/1/2^*-FGFS}}{\Delta_{M/M/1/2^*-PS}}\leq \frac{4}{3}.
\label{eq:mm12-star-comparison}
\end{equation}
\label{prop:mm12-star-comparison}
\end{proposition}

From this result, we conclude that the AAoI of the M/M/1/2 queue with preemption and under PS discipline is always smaller than the
AAoI of the M/M/1/2 queue with preemption and under FGFS discipline. Besides, we also conclude that $\Delta_{M/M/1/2^*-FGFS}$ is, 
at most, $4/3$ times worse than $\Delta_{M/M/1/2^*-PS}$.

The authors in \cite{costa2016age} showed that the AAoI of the M/M/1/2 queue with preemption and under FGFS is smaller than
the AAoI without preemption and under FGFS. This implies that, for the FGFS, when the maximum number of packets in the queue is two, preemption of packets leads to a performance improvement. In the following result, we study if such a performance improvement is also achieved when we consider the PS queue instead of the FGFS queue. Its proof is presented in Appendix~\ref{app:prop:mm12-mm12-star-comparison}.

\begin{proposition}
We have that
\begin{equation}
1\leq \frac{\Delta_{M/M/1/2-PS}}{\Delta_{M/M/1/2^*-PS}}\leq \frac{5}{3}.
\label{eq:mm12-star-comparison}
\end{equation}
\label{prop:mm12-mm12-star-comparison}
\end{proposition}

From this result, we derive that the aforementioned property shown in \cite{costa2016age} about the preemption of the FGFS when the maximum 
number of packets is two also holds when we consider the PS discipline.

\subsection{The M/M/1/2$^{**}$ Queue}
\label{sec:mm12star2}
We now consider the M/M/1/2 queue with preemption to the oldest packet. In this system,  when a new packet arrives to the system and there are two packets in the queue, unlike in the previous section, the oldest packet in the queue is replaced by the incoming one. We denote this queueing model as the M/M/1/2$^{**}$ queue.

We now present the expression of the AAoI of the M/M/1/2$^{**}$. The proof is presented in Appendix~\ref{app:prop:mm12-star2-ps}.

\begin{proposition}
The AAoI of the M/M/1/2$^{**}$-PS queue is
\begin{multline}
\Delta_{M/M/1/2^{**}-PS}= \\\frac{2\lambda^6+11\lambda^5\mu+25\lambda^4\mu^2+29\lambda^3\mu^3+22\lambda^2\mu^4+10\lambda\mu^5+2\mu^6}
{2\lambda\mu(\lambda+\mu)^3(\lambda^2+\mu^2+\lambda\mu)}.
\label{eq:mm12-ps-star2}
\end{multline}
\label{prop:mm12-star2-ps}
\end{proposition}

We aim to compare the AAoI of the M/M/1/2$^{**}$-PS queue with the AAoI of the M/M/1/2$^{**}$-FGFS queue. We present the expression of the latter in the following proposition. The proof is available in Appendix~\ref{app:prop:mm12-star2-fgfs}.

\begin{proposition}
The AAoI of the M/M/1/2$^{**}$-FGFS queue is
\begin{multline}
\Delta_{M/M/1/2^{**}-FGFS}= \\\frac{\lambda^6+6\lambda^5\mu+14\lambda^4\mu^2+15\lambda^3\mu^3+11\lambda^2\mu^4+5\lambda\mu^5+\mu^6}
{\lambda\mu(\lambda+\mu)^3(\lambda^2+\mu^2+\lambda\mu)}.
\label{eq:mm12-fgfs-star2}
\end{multline}
\label{prop:mm12-star2-fgfs}
\end{proposition}

We focus on the ratio $\frac{\Delta_{M/M/1/2^{**}-FGFS}}{\Delta_{M/M/1/2^{**}-PS}}$.  In the following result, as in the previous section, we provide a lower-bound and an upper-bound of this ratio. The proof is given in Appendix~\ref{app:prop:mm12-star2-comparison}.

\begin{proposition}
We have that
\begin{equation}
1\leq \frac{\Delta_{M/M/1/2^{**}-FGFS}}{\Delta_{M/M/1/2^{**}-PS}}\leq 1.0731.
\label{eq:mm12-star2-comparison}
\end{equation}
\label{prop:mm12-star2-comparison}
\end{proposition}

From this result, we conclude that the AAoI of the M/M/1/2$^{**}$-PS queue is slightly smaller than the AAoI of the   M/M/1/2$^{**}$-FGFS queue.  Furthermore, we also show that $\Delta_{M/M/1/2^{**}-FGFS}$ is, at most, $1.0731$ worse than  $\Delta_{M/M/1/2^{**}-PS}$.

We now want to compare the two different systems with preemption and under the PS discipline.  As a matter of fact, we aim to compare the AAoI of the M/M/1/2$^*$ queue and the AAoI of the M/M/1/2$^{**}$ queue under the PS discipline.  In the following result, we study if discarding the oldest packet agaisnt the newest one leads to a performance improvement. We present the proof in Appendix~\ref{app:prop:mm12-star1-2-ps-comparison}.

\begin{proposition}
We have that
\begin{equation}
1\leq \frac{\Delta_{M/M/1/2^{*}-PS}}{\Delta_{M/M/1/2^{**}-PS}}\leq \frac{3}{2}.
\label{eq:mm12-star1-2-ps-comparison}
\end{equation}
\label{prop:mm12-star1-2-ps-comparison}
\end{proposition}

From this result, we see that, indeed, preemption replacing the oldest packet in the queue leads to a performance improvement against preemption replacing the newest packet. And that $\Delta_{M/M/1/2^{*}-PS}$ is, at most, $\frac{3}{2}$ worse than $\Delta_{M/M/1/2^{**}-PS}$.  Moreover, from Proposition~\ref{prop:mm12-mm12-star-comparison} and Proposition~\ref{prop:mm12-star1-2-ps-comparison} we derive the following result.  

\begin{corollary}
We have that
\begin{equation}
1\leq \frac{\Delta_{M/M/1/2-PS}}{\Delta_{M/M/1/2^{**}-PS}}\leq \frac{5}{2}.
\label{eq:mm12-ps-mm12-star2-ps-comparison}
\end{equation}
\label{cor:mm12-ps-mm12-star2-ps-comparison}
\end{corollary}

We conclude that the AAoI of the M/M/1/2 queue under the PS discipline and with optimal preemption can be, at most, 2.5 better than the one without preemption. 

\subsection{Comparison with the M/M/1/1 Queue}
\label{sec:comparison-mm11}
In this section our goal will be comparing the M/M/1/1 queue with the M/M/1/2$^{**}$ queue under the PS discipline.  In the M/M/1/1 queue system, the maximum number of packets that can be stored in the queue is one. Besides, in the M/M/1/1 queue, when a new packet arrives and there is already a packet in the system, the incoming packet is discarded.

It is shown in \cite{costa2016age} the following result that characterizes the AAoI of the M/M/1/1 queue:
\begin{equation}
\Delta_{M/M/1/1}= \frac{2\lambda^2+2\lambda\mu+\mu^2}{\lambda\mu(\lambda+\mu)}.
\label{eq:mm11}
\end{equation}

We now focus on the ratio $\frac{\Delta_{M/M/1/1}}{\Delta_{M/M/1/2^{*}-PS}}$. We give a lower and an upper bound for the ratio in the following result.  The proof is available in Appendix~\ref{app:mm11-mm12-star-ps-comparison}.

\begin{proposition}
We have that
\begin{equation}
1\leq \frac{\Delta_{M/M/1/1}}{\Delta_{M/M/1/2^{*}-PS}}\leq \frac{4}{3}.
\label{eq:mm11-mm12-star-ps-comparison}
\end{equation}
\label{prop:mm11-mm12-star-ps-comparison}
\end{proposition}

From this result, we see that the AAoI of the M/M/1/2$^{*}$ queue under the PS discipline is smaller than the AAoI of the M/M/1/1 queue (and, as a consequence, AAoI of the M/M/1/2$^{**}$ queue under the PS discipline is smaller than the AAoI of the M/M/1/1 queue). Moreover, we conclude that $\Delta_{M/M/1/1}$ will be, at most, 2 times worse than $\Delta_{M/M/1/2^{**}-PS}$, i.e., 
$$
1\leq \frac{\Delta_{M/M/1/1}}{\Delta_{M/M/1/2^{**}-PS}}\leq 2.
$$

We now compare the AAoI of the M/M/1/1 queue with the AAoI of the M/M/1/2 queue. We present the proof in Appendix~\ref{app:prop:mm11-mm12-star2-ps-comparison2}.

\begin{proposition}
We have that
\begin{equation}
0.9641\leq \frac{\Delta_{M/M/1/2-PS}}{\Delta_{M/M/1/1}}\leq \frac{5}{4}.
\label{eq:mm11-mm12-star2-ps-comparison2}
\end{equation}
\label{prop:mm11-mm12-star2-ps-comparison2}
\end{proposition}

From this result, we conclude that, when $\lambda\in[0,\mu]$, then we have that $\Delta_{M/M/1/1}\geq \Delta_{M/M/1/2-PS}$, whereas when $\lambda\in[\mu,\infty)$, we have that $\Delta_{M/M/1/2-PS}\geq \Delta_{M/M/1/1}$. Besides, the AAoI of the M/M/1/1 queue can be, at most, 5/4 times better than the AAoI of the M/M/1/2-PS queue and  the AAoI of the M/M/1/2-PS queue can be, at most, $1/0.9641\approx1.0372$ times better than the AAoI of the M/M/1/1 queue.

\subsection{Comparison with the M/M/1/1$^*$ Queue}
\label{sec:comparison-mm11star}
We now want to extend the results of the previous section to the M/M/1/1$^*$ queue.  In the M/M/1/1$^*$ system we have preemption, when a new packet arrives while there is a packet in the queue, the packet will be replaced by the incoming one. 

In \cite{yates2018age} it is shown that the expression of $\Delta_{M/M/1/1^*}$ is
\begin{equation}
\Delta_{M/M/1/1^*}= \frac{\lambda+\mu}{\lambda\mu}.
\label{eq:mm11-star}
\end{equation}

Now, we compare the AAoI of the M/M/1/1$^*$ queue with the AAoI of the M/M/1/2$^{**}$-PS queue. In order to that we give the following result. The proof is available in Appendix~\ref{app:prop:mm11-star-mm12-star2-ps-comparison}.

\begin{proposition}
We have that
\begin{equation}
1\leq \frac{\Delta_{M/M/1/2^{**}-PS}}{\Delta_{M/M/1/1^*}}\leq 1.0788.
\label{eq:mm11-star-mm12-star2-ps-comparison}
\end{equation}
\label{prop:mm11-star-mm12-star2-ps-comparison}
\end{proposition}

\section{The M/M/1 Queue}
\label{sec:mm1}

We now analyze the AAoI of  the M/M/1 queue with the PS discipline. For this case, we assume that $\rho<1$ to ensure stability. 
Our first result of this section consists of establishing a lower-bound of $\Delta_{M/M/1-PS}$. 
Its proof is available in Appendix~\ref{app:lem:mm1-lowerbound}. 

\begin{lemma}
$$
\Delta_{M/M/1-PS}> \frac{\mu-\lambda}{\lambda\mu}.
$$
\label{lem:mm1-lowerbound}
\end{lemma}

It is shown in \cite{kaul2012real} that 
\begin{equation}
\Delta_{M/M/1-FGFS}=\frac{1}{\mu}\left(1+\frac{1}{\rho}+\frac{\rho^2}{1-\rho}\right),
\label{eq:mm1-fgfs-aaoi}
 \end{equation}
 which is clearly unbounded from above when $\lambda\to 0$. This result implies that, when we consider the M/M/1-FGFS model, the arrival rate that minimizes the mean number of customers does not minimize the AAoI. 
Using Lemma~\ref{lem:mm1-lowerbound}, we now show that this property also holds when we consider the M/M/1-PS model.

\begin{proposition}
When $\lambda\to 0$, $ \Delta_{M/M/1-PS}$ is unbounded from above.
\label{prop:mm1-unbounded-lambda-zero}
\end{proposition}

\begin{proof}
From Lemma~\ref{lem:mm1-lowerbound}, the desired result follows noting that, when $\lambda\to 0$, $\frac{\mu-\lambda}{\lambda\mu}$ tends to infinity. 
\end{proof}

We have tried to provide an 
explicit expression of $\Delta_{M/M/1-PS}$ using the Stochastic Hybrid System (SHS) technique. Unfortunately, the
derived expression are extremely difficult and, therefore, we did not succeed in characterizing $\Delta_{M/M/1-PS}$. After extensive numerical experiments, we conjecture that the AAoI of the M/M/1-PS queue has a similar form as  $\Delta_{M/M/1-FGFS}$. To be more precise, we now present our conjecture.
\begin{conjecture}
\begin{equation}
\Delta_{M/M/1-PS}=\frac{1}{\mu}\left(\frac{1}{\rho}+1+C(\rho)\right),
\label{eq:mm1ps-conjec}
\end{equation}
where $\lim_{\rho\to 0}C(\rho)=0$, $\lim_{\rho\to 1}C(\rho)=+\infty$ and $0 \leq C(\rho)\leq \frac{\rho^2}{1-\rho}$ for all $\rho\in (0,1)$. 
Moreover, when $\rho$ is large enough,
\begin{equation}
\frac{(\rho-0.5)^3}{1- \rho} \leq C(\rho)\leq \frac{0.75\rho}{(1- \rho)^{\frac{1}{2}}}.
\label{eq:c-lb-ub}
\end{equation}
\label{conject:mm1ps}
\end{conjecture}

We remark that, if Conjecture~\ref{conject:mm1ps} holds, then it follows that, when $\lambda$ is large enough,
$$
\Delta_{M/M/1-PS}\geq \frac{1}{\mu}\left(\frac{1}{\rho}+1+\frac{(\rho-0.5)^3}{1- \rho}\right).
$$

Let us note that the rhs of the above expression tends to infinity when $\rho\to 1$.
Therefore, we conclude that,
if Conjecture~\ref{conject:mm1ps} holds, when $\rho\to 1$, $\Delta_{M/M/1-PS}$ tends to infinity.

We know from \eqref{eq:mm1-fgfs-aaoi} that, when $\rho\to 1$, the AAoI of the M/M/1-FGFS queue tends to infinity, therefore
the load that maximizes the throughput does not optimize the AAoI for this model. Now, we conclude that, if Conjecture~\ref{conject:mm1ps} holds, then the aforementioned property is verified also for the M/M/1-PS queue. 

In the following result, we study the value of the ratio $\frac{\Delta_{M/M/1-FGFS}}{\Delta_{M/M/1-PS}}$. The proof is available in Appendix~\ref{app:prop:mm1-ratio}.

\begin{proposition}
If Conjecture~\ref{conject:mm1ps} holds, then
$$
1\leq \frac{\Delta_{M/M/1-FGFS}}{\Delta_{M/M/1-PS}} \leq +\infty.
$$
\label{prop:mm1-ratio}
\end{proposition}

\section{Multiple Sources}
\label{sec:multiple-sources}

In this section we will expand our analysis from single-source systems to multiple source systems.  In fact, we consider that there are two sources sending update packets through the transmission channel to the monitor following Poisson processes. The rate at which updates of source 1 are sent is $\lambda_1$ and of source 2 is $\lambda_2$.  The service rate is exponentially distributed with rate $\mu$ for the updates from any source. 

We aim to analyze the impact of $\lambda_2$ on the AAoI of the updates of source one under the different queueing disciplines. In our numerical analysis we consider that $\mu=1$ and we represent with a solid line the AAoI of the M/M/1-PS queue, with a dashed line the AAoI of the M/M/1-FGFS and with a dotted line the AAoI of the M/M/1/1$^*$ queue. In Figure~\ref{fig:multiple-lambda1small}, we consider that $\lambda_1=0.1$ and $\lambda_2$ varying from $0.001$ to $0.05$.
We observe that the influence of $\lambda_2$ is very similar for PS and FGFS (and that PS outperforms FGFS for all $\lambda_2$), but the AAoI of the M/M/1/1$^*$ queue increases dramatically with $\lambda_2$. Indeed, when $\lambda_2$ is close to zero, the AAoI of the M/M/1/1$^*$ queue is the smallest one, but further numerical experiments show that it tends to infinity when $\lambda_2$ grows large much faster than FGFS and PS. Therefore, we conclude that, when $\lambda_1$ is small, the presence of other sources has a very negative impact in the AAoI of the M/M/1/1$^*$ compared to the AAoI of PS and FGFS. 

\begin{figure}[t!]
\centering
\includegraphics[width=\columnwidth,clip=true,trim=0pt 200pt 0pt 220pt]{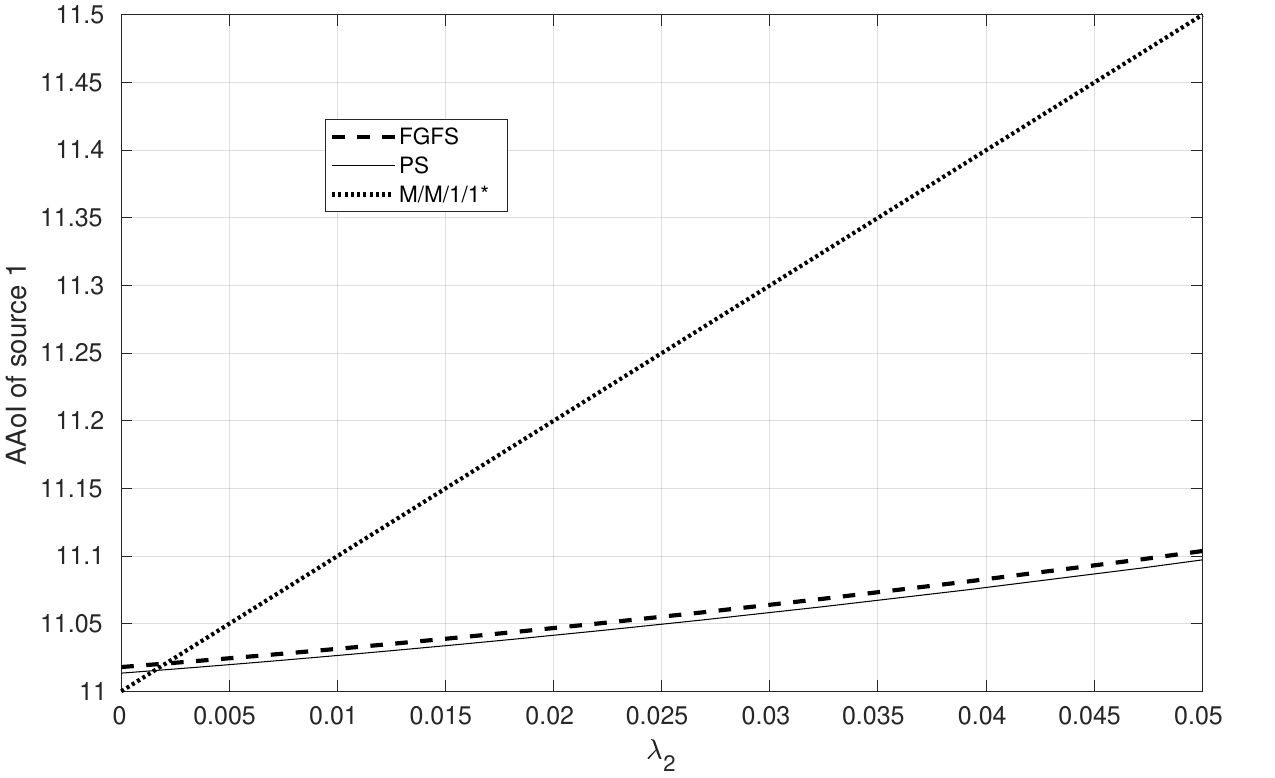}
\caption{Comparison of the AAoI of source 1 in a M/M/1 queue with divers queueing disciplines when $\lambda_2$ changes from 0.001 to 0.05 and $ \lambda_1=0.1$. }
\label{fig:multiple-lambda1small}
\end{figure}

We now aim to study whether the conclusions obtained for $\lambda_1$ small extend to instances where $\lambda_1$ is 
large. To this end, we consider $\lambda_1=5$ and $\lambda_2$ varying from 0.001 to $10^3$. In Figure~\ref{fig:multiple-lambda1large}, we represent the values of the AAoI we have obtained in our numerical analysis. We observe that the influence of $\lambda_2$ is again very similar for PS and FGFS and when $\lambda_2$ is close to zero the AAoI of the M/M/1/1$^*$ queue is the smallest one. However, when $\lambda_2$ is large, the AAoI of the M/M/1/1$^*$ queue is not worse than that of PS and FGFS; they equal, in fact, the same value. As a result, we conclude that, when $\lambda_1$ is large, the presence of a different source does not have a very negative impact on the performance of the M/M/1/1$^*$ queue compared to the PS and FGFS. 

\begin{figure}[t!]
\centering
\includegraphics[width=\columnwidth,clip=true,trim=0pt 200pt 0pt 220pt]{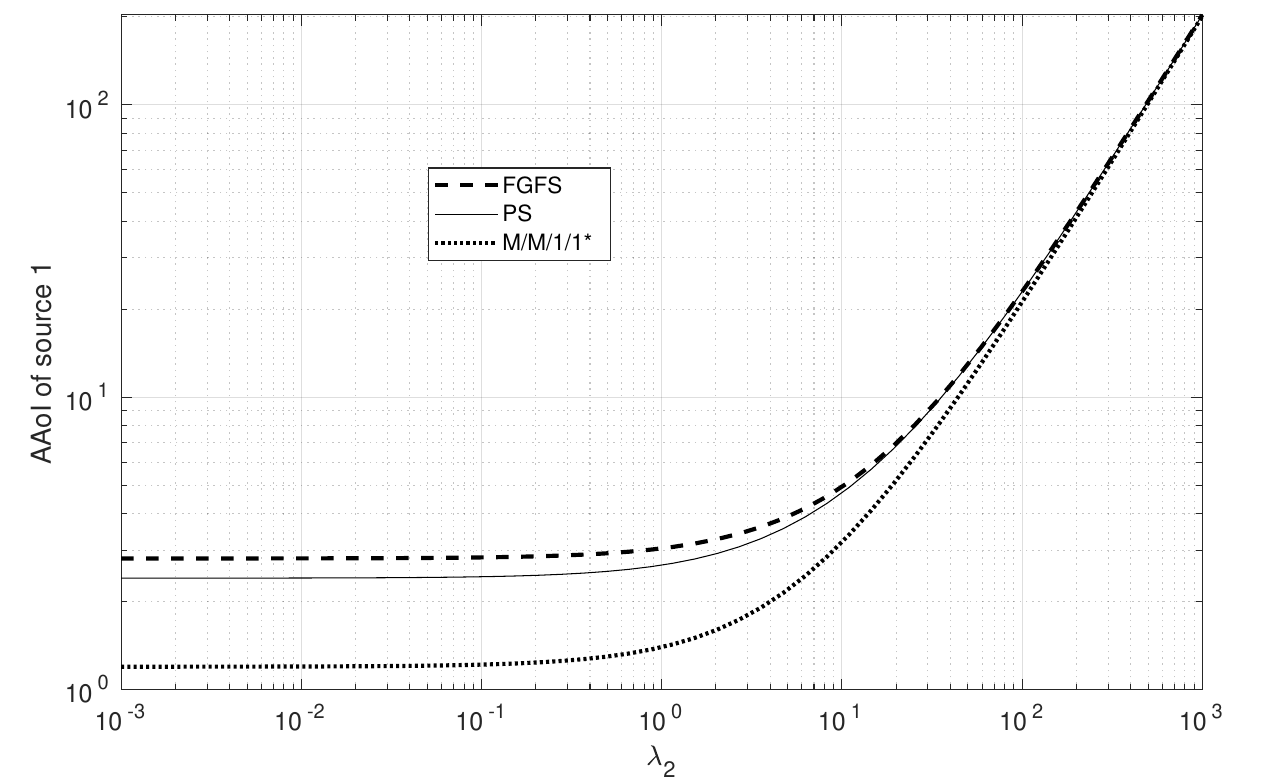}
\caption{Comparison of the AAoI of the source 1 in a M/M/1 queue with divers queueing disciplines when $\lambda_2$ changes from 0.001 to $10^3$ and $\lambda_1=5$.}
\label{fig:multiple-lambda1large}
\end{figure}

We now focus on the AAoI of both sources. To this aim, we analyze the evolution of the sum of the AAoI of both sources as a function of the arrival rate of one of them. We consider $\mu=1$ in these experiments and same queueing disciplines as in Figure~\ref{fig:multiple-lambda1small} and Figure~\ref{fig:multiple-lambda1large}. In Figure~\ref{fig:multiple2-lambda1large}, we set $\lambda_1=0.1$ and we plot the AAoI of both sources when $\lambda_2$ changes from 0.001 to 30. We observe that, in this case, the AAoI of the M/M/1/1$^*$ queue is larger than the AAoI of PS and FGFS. Therefore, we conclude that, when the arrival rate of one of the source is low, PS and FGFS outperform M/M/1/1$^* $. However, in Figure~\ref{fig:multiple2-lambda1small}, we consider $\lambda_1=5$ to analyze whether the aforementioned conclusions extend to the instance where $\lambda_1$ is large. We observe that, when $\lambda_2$ is small, the AAoI of the M/M/1/1$^*$ queue is again larger than the AAoI of PS and FGFS, whereas when $\lambda_2$ is large, the AAoI of the M/M/1/1$^*$ queue is smaller. We conclude that, when the arrival rate of both sources is large, it is preferable from the perspective of the AAoI the M/M/1/$^*$ queue and, in the rest of the cases, PS or FGFS are preferable. 

\begin{figure}[t!]
\centering
\includegraphics[width=\columnwidth,clip=true,trim=0pt 200pt 0pt 220pt]{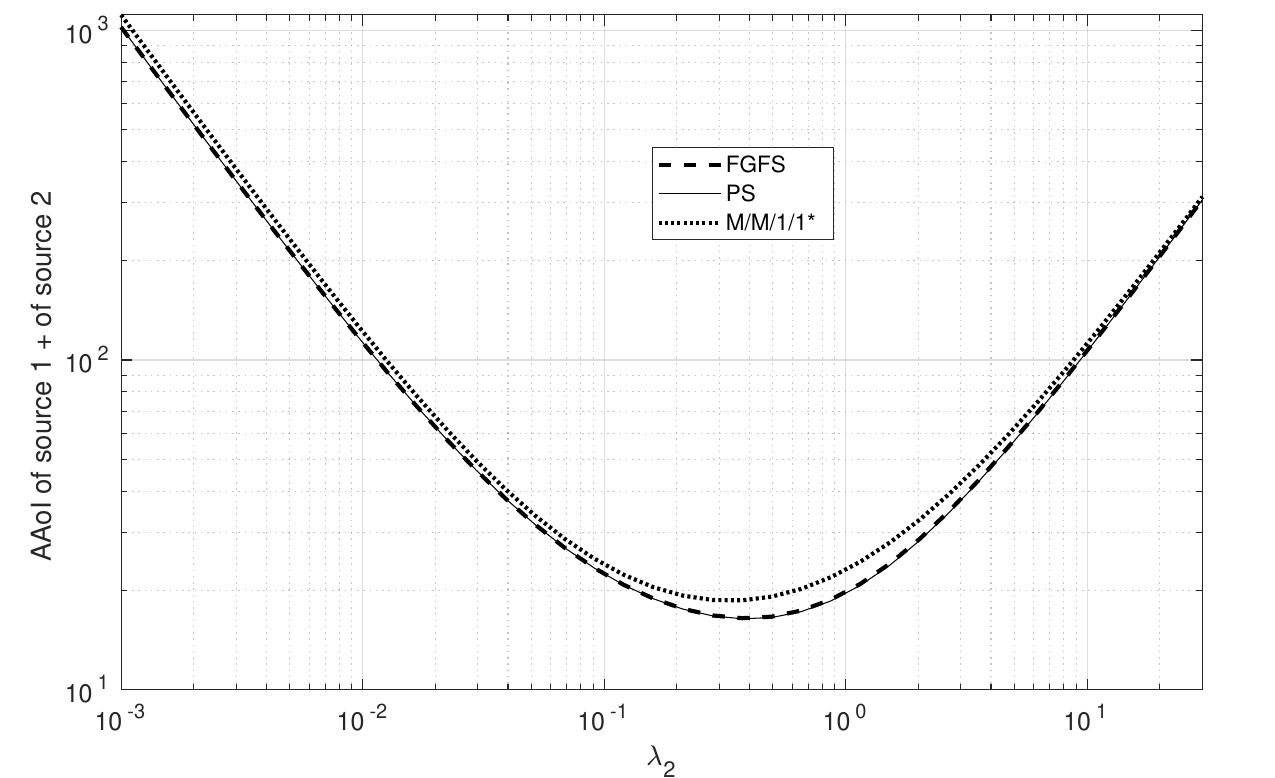}
\caption{Comparison of the AAoI of source 1 plus the AAoI of source 2 in a M/M/1 queue with divers queueing disciplines when $\lambda_2$ changes from 0.001 to 30 and $ \lambda_1=0.1$. }
\label{fig:multiple2-lambda1small}
\end{figure}

\begin{figure}[t!]
\centering
\includegraphics[width=\columnwidth,clip=true,trim=0pt 200pt 0pt 220pt]{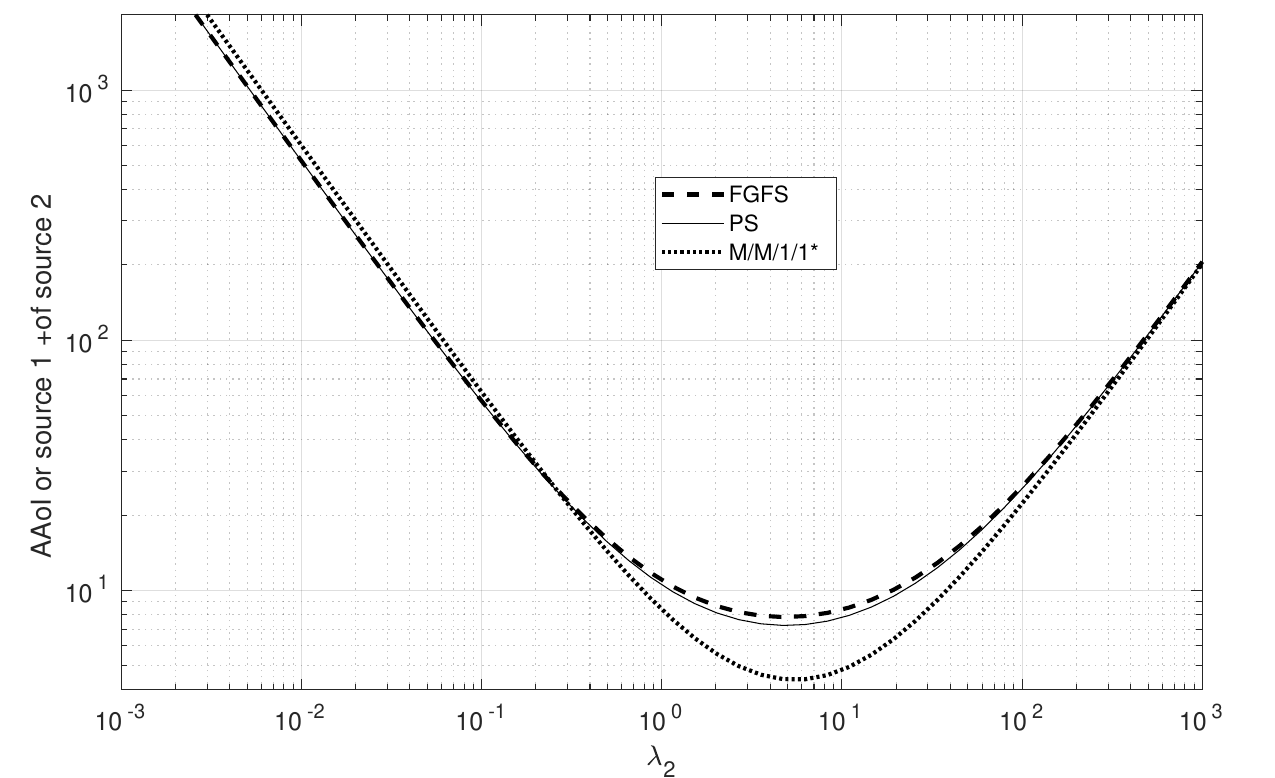}
\caption{Comparison of the AAoI of source 1 plus the AAoI of source 2 in a M/M/1 queue with divers queueing disciplines when $\lambda_2$ changes from 0.001 to $10^3$ and $\lambda_1=5$.}
\label{fig:multiple2-lambda1large}
\end{figure}

\section{Conclusions}
\label{sec:conclusions}

In this paper, we investigated the average AoI in a system composed of sources sending status updates to a monitor through a Processor Sharing (PS) queue. We considered the  single source M/M//2 queue with and without preemption, and derived a closed-form expression of the average AoI by  making use of  SHS tool. We compared analytically our results to the FGFS discipline. The results of this work is consistent with \cite{bedewy2016optimizing} since we provide analytical results that show that disciplines without queueing have good AoI performance.

We then extended the analysis to the M/M/1 queue with one and two sources. We solved numerically the  equations resulting from the SHS framework and compared the obtained results with the FGFS and M/M/1/1*, which is known to have good AoI  performance. Our results show that the PS discipline can outperform the M/M/1/1* queue in some cases.

%
%
%
%


%
\appendices
\section{Proof of Proposition~\ref{prop:mm12-ps}}
\label{app:prop:mm12-ps}

We use the SHS methodology \cite{yates2018age} to characterize the AAoI of the M/M/1/2-PS queue. 
The SHS technique is formed by a couple $(x,q)$ where $x$ is a continuous state and $q$ a discrete state. For this model, the discrete state belongs to the continuous time Markov chain illustrated in Figure~\ref{fig:mm12}, where each state represents the number of packets in the queue; the continuous state is a vector $\mathbf x(t)=[x_0(t)\ x_1(t)\ x_2(t)]$ where $x_0(t)$ is the Age of Information at time $t$ and $x_i(t)$ is the age of the i-th packet in the queue. We also define $b_0=[1,0,0],$ $b_1=[1,1,0]$ and $b_2=[1,1,1]$ that represent which are the packets whose age increases at rate one for each of the states of the Markov chain of Figure~\ref{fig:mm12}.

The second column of Table~\ref{tab:mm12-sh} represents the rate at which transitions of the Markov chain occur. The steady-state distribution of this Markov chain is clearly 
$$
\pi_{i}=\frac{\rho^i}{1+\rho+\rho^2}, i=0,1,2.
$$
We now describe each of the transitions of Table~\ref{tab:mm12-sh}.\\
\begin{itemize}
\item[$l=0$] A new packet arrives when the queue is empty. This occurs with rate $ \lambda$. For this case, the age of the monitor does not change and the age of the first packet in the queue is equal to zero, i.e., $x_1'=0$.
\item[$l=1$] There is one packet in the queue and it is served, which occurs with rate $\mu$. For this case, the age of the monitor is replace by the age of the packet in service, i.e., $x_0'=x_1$.
\item[$l=2$] A new packet arrives when there is another packet in the queue. This occurs with rate $\lambda$. For this case, the age of the monitor and of the packet that was being served in the queue do not change. However, the age of the second packet is equal to zero.
\item[$l=3$] There are two packets in the system and the packet that arrived first is served. This occurs with rate $\mu/2$. For this
case, the age of the monitor changes to the age of the packet that has been served, i.e., $x_0'=x_1$. Besides, the packet that stays in the queue has become the freshest of the packets in the queue and, therefore, $x_1'=x_2$.
\item[$l=4$] There are two packets in the system and the packet that last arrived is served. This occurs with rate $ \mu/2$. For this case, the age of the monitor is replaced by the age of the last arrived packet, i.e., $x_0'=x_2$. Besides, the packet that stays in the queue is obsolete and, therefore, we replace it by a fake update with the same age of the served packet, i.e., $x_1'=x_2$.
\item[$l=5$] There are two packets in the system and a new packet arrives. This occurs with rate $\lambda$. For this case, the new incoming packet is discarded, therefore the age of the monitor and of the packets in the queue does not change. 
\end{itemize}

\begin{figure}[t!]
	\centering
	\begin{tikzpicture}[]
	\node[style={circle,draw}] at (0,0) (1) {$0$};
	\node[style={circle,draw}] at (3,0) (2) {$1$};
	\node[style={circle,draw}] at (6,0) (3) {$2$};
	\node[] at (7.2,0) (X) {$l=5$};
	\draw[->] (1) edge [bend left] node[above] {$l=0$} (2);
	\draw[<-] (1) edge [bend right] node[below] {$l=1$} (2);
	\draw[->] (3) to [out=340,in=20,looseness=8] (3) ;
	\draw[->] (2) edge [bend left] node[above] {$l=2$} (3);
	\draw[<-] (2) edge [bend right] node[below] {$l=3$} (3);
	\draw[<-] (2) edge [out=280, in=260] node[below] {$l=4$} (3);
	\end{tikzpicture}  
	\caption{The Markov chain under consideration in the proof of Proposition~\ref{prop:mm12-ps}.}
	\label{fig:mm12}
\end{figure}
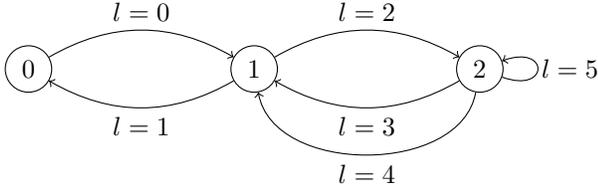

\begin{table}[t!]
\begin{center}
\begin{tabular}{|c c c c |c|} 
 \hline
 $l$ & $\lambda^{(l)}$ & $q\longrightarrow q'$ & $x\longrightarrow x'=A_lx$ & $v_{ql}A_l$ \\ \hline
 0 & $\lambda$ & $0\longrightarrow 1$ & $[x_0,0,0]\longrightarrow [x_0,0,0]$  & $[v_{00},0,0]$\\ 
 \hline
 1 & $\mu$ & $1\longrightarrow 0$ & $[x_0,x_1,0]\longrightarrow [x_1,0,0]$ & $[v_{11},0,0]$ \\
 \hline
 2 & $\lambda$ & $1\longrightarrow 2$ & $[x_0,x_1,0]\longrightarrow [x_0,x_1,0]$ & $[v_{10},v_{11},0]$\\
 \hline
 3 & $\frac{\mu}{2}$ & $2\longrightarrow 1$ & $[x_0,x_1,x_2]\longrightarrow [x_1,x_2,0]$ & $[v_{21},v_{22},0]$\\
 \hline
 4 & $\frac{\mu}{2}$ & $2\longrightarrow 1$ & $[x_0,x_1,x_2]\longrightarrow [x_2,x_2,0]$ & $[v_{22},v_{22},0]$\\  
 \hline
  5 & $\lambda$ & $2\longrightarrow 2$ & $[x_0,x_1,x_2]\longrightarrow [x_0,x_1,x_2]$ & $[v_{20},v_{21},v_{22}]$\\ [1ex] 
 \hline
\end{tabular}
\end{center}
\caption{The SHS table under consideration in the proof of Proposition~\ref{prop:mm12-ps}.}
\label{tab:mm12-sh}
\end{table}




We apply (35a) of \cite{yates2018age} for our case and we obtain:
\begin{align*}
[v_{00},v_{01},v_{02}]\lambda=&b_0\pi_0 + \mu [v_{11},0,0]\\
[v_{10},v_{11},v_{12}](\lambda + \mu)=&b_1\pi_1 + \lambda [v_{00},0,0]+\frac{\mu}{2} [v_{21},v_{22},0]\\&+\frac{\mu}{2}[v_{22},v_{22},0]\\
[v_{20},v_{21},v_{22}](\lambda + \mu)=&b_2\pi_2 + \lambda [v_{10},v_{11},0]+\lambda [v_{20},v_{21},v_{22}].\\
\end{align*}

From Theorem 4 of  \cite{yates2018age}, we know that, if there exists a non negative solution of the above system of equations, then the  AAoI of this model is given by $v_{00}+v_{10}+v_{20}$. 

We solve the above system of equations and we get
\begin{align*} 
v_{00}&=\frac{\mu (\lambda + \mu)}{\lambda (\lambda^2+\mu^2+\lambda\mu)}\\
v_{10}&=\frac{3\lambda^2+2\mu^2+4\lambda\mu}{2(\lambda +\mu)(\lambda^2+\mu^2+\lambda\mu)}\\
v_{11}&=\frac{\lambda}{\lambda^2+\lambda \mu+\mu^2}\\
v_{20}&=\frac{5\lambda^3+6\lambda^2\mu+2\lambda\mu^2}{2\mu(\lambda +\mu)(\lambda^2+\mu^2+\lambda\mu)}\\
v_{21}&=\frac{2\lambda^2}{\mu(\lambda^2+\lambda \mu+\mu^2)}\\
v_{22}&=\frac{\lambda^2}{\mu(\lambda^2+\lambda \mu+\mu^2)}
 \end{align*}

Therefore,
\begin{align*}
\Delta_{M/M/1/2-PS}&=v_{00}+v_{10}+v_{20}
\\&=\frac{5\lambda^4+9\lambda^3\mu+8\lambda^2\mu^2+6\lambda\mu^3+2\mu^4}{2\lambda\mu(\lambda +\mu)(\lambda^2+\mu^2+\lambda\mu)},
\end{align*}
where the last equality has been obtained simplifying the derived expression. 
And the desired result follows.

\section{Proof of Proposition~\ref{prop:mm12-comparison}}
\label{app:prop:mm12-comparison}

We first note that \eqref{eq:mm12-fgfs} can be written as follows
$$
\frac{3\lambda^4+5\lambda^3\mu+4\lambda^2\mu^2+3\lambda\mu^3+\mu^4}{\lambda\mu(\lambda+\mu)(\lambda^2+\lambda\mu+\mu^2)}.
$$
As a result,
\begin{equation}
\frac{\Delta_{M/M/1/2-FGFS}}{\Delta_{M/M/1/2-PS}}=
2\frac{3\lambda^4+5\lambda^3\mu+4\lambda^2\mu^2+3\lambda\mu^3+\mu^4}{5\lambda^4+9\lambda^3\mu+8\lambda^2\mu^2+6\lambda\mu^3+2\mu^4}.
\label{eq:mm12-comparison-aux1}
\end{equation}

Thus, taking into account that the rhs of \eqref{eq:mm12-comparison-aux1} tends to 1 when $\lambda\to 0$ and to 1.2 when $\lambda\to \infty$, 
the desired result follows if we show that the rhs of \eqref{eq:mm12-comparison-aux1} is increasing with $\lambda$, which we proof in the following result.
\begin{lemma}
The rhs of \eqref{eq:mm12-comparison-aux1} is an increasing function of $\lambda$ for all $\mu>0$.
\label{lem:mm12-comparison-increasing-lambda}
\end{lemma}

 \begin{proof}
We compute the derivative of $\frac{3\lambda^4+5\lambda^3\mu+4\lambda^2\mu^2+3\lambda\mu^3+\mu^4}{5\lambda^4+9\lambda^3\mu+8\lambda^2\mu^2+6\lambda\mu^3+2\mu^4}$ with respect to $\lambda$ and it results
   
    \begin{multline*}\frac{12\lambda^3+15\lambda^2\mu+8\lambda\mu^2+3\mu^3}{5\lambda^4+9\lambda^3\mu+8\lambda^2\mu^2+6\lambda\mu^3+2\mu^4}\\-(20\lambda^3+27\lambda^2\mu+16\lambda\mu^2+6\mu^3)\\\frac{(3\lambda^4+5\lambda^3\mu+4\lambda^2\mu^2+3\lambda\mu^3+\mu^4)}{(5\lambda^4+9\lambda^3\mu+8\lambda^2\mu^2+6\lambda\mu^3+2\mu^4)^2}.
\end{multline*}
    The above expression is positive if and only if
    \begin{multline*}(12\lambda^3+15\lambda^2\mu+8\lambda\mu^2+3\mu^3)\\(5\lambda^4+9\lambda^3\mu+8\lambda^2\mu^2+6\lambda\mu^3+2\mu^4)\\>(20\lambda^3+27\lambda^2\mu+16\lambda\mu^2+6\mu^3)\\(3\lambda^4+5\lambda^3\mu+4\lambda^2\mu^2+3\lambda\mu^3+\mu^4),\end{multline*}
which simplifying we obtain that
    $$2\lambda^6\mu+8\lambda^5\mu^2+12\lambda^4\mu^3+10\lambda^3\mu^4+3\lambda^2\mu^5 > 0,$$
which is true for all $\mu>0$. And the desired result follows.
\end{proof}

\section{Proof of Proposition~\ref{prop:mm12-star-ps}}
\label{app:prop:mm12-star-ps}
The proof is similar to that of Proposition~\ref{prop:mm12-ps}. In fact, we formulate the same SHS approach with the exception of transition $l=5$, which is described in Table~\ref{tab:mm12-sh-star}.

We again apply (35a) of \cite{yates2018age} and we get the following system of equations:
\begin{align*}
[v_{00},v_{01},v_{02}]\lambda=&b_0\pi_0 + \mu [v_{11},0,0]\\
[v_{10},v_{11},v_{12}](\lambda + \mu)=&b_1\pi_1 + \lambda [v_{00},0,0]+\frac{\mu}{2} [v_{21},v_{22},0]\\&+\frac{\mu}{2}[v_{22},v_{22},0]\\
[v_{20},v_{21},v_{22}](\lambda + \mu)=&b_2\pi_2 + \lambda [v_{10},v_{11},0]+\lambda [v_{20},v_{21},0].\\
\end{align*}

The solution of this system of linear equations is
\begin{align*} 
v_{00}&=\frac{3\lambda^2\mu^2+3\lambda\mu^3+\mu^4}{\lambda(\lambda+\mu)^2(\lambda^2+\mu^2+\lambda\mu)}\\
v_{10}&=\frac{\lambda^4+7\lambda^3\mu+13\lambda^2\mu^2+8\lambda\mu^3+2\mu^4}{2(\lambda+\mu)^3(\lambda^2+\mu^2+\lambda\mu)}\\
v_{11}&=\frac{2\lambda^2\mu+\lambda\mu^2}{(\lambda+\mu)^2(\lambda^2+\mu^2+\lambda\mu)}\\
v_{20}&=\frac{3\lambda^5+11\lambda^4\mu+15\lambda^3\mu^2+14 \lambda^2\mu^3+8\lambda\mu^4+2\mu^5}
{2\lambda\mu(\lambda+\mu)^2(\lambda^2+\mu^2+\lambda\mu)}\\
v_{21}&=\frac{\lambda^4+4\lambda^3\mu+2\lambda^2\mu^2}{\mu(\lambda+\mu)^2(\lambda^2+\mu^2+\lambda\mu)}\\
v_{22}&=\frac{\lambda^2}{(\lambda+\mu)(\lambda^2+\mu^2+\lambda\mu)}.
 \end{align*}

According to Theorem 4 of  \cite{yates2018age}, the desired value is obtained by summing $v_{00}$, $v_{10}$ and $v_{20}$. 
And the desired result follows.

\begin{table}[t!]
\begin{center}
\begin{tabular}{|c c c c |c|} 
 \hline
 $l$ & $\lambda^{(l)}$ & $q\longrightarrow q'$ & $x\longrightarrow x'=A_lx$ & $v_{ql}A_l$ \\ \hline
 0 & $\lambda$ & $0\longrightarrow 1$ & $[x_0,0,0]\longrightarrow [x_0,0,0]$  & $[v_{00},0,0]$\\ 
 \hline
 1 & $\mu$ & $1\longrightarrow 0$ & $[x_0,x_1,0]\longrightarrow [x_1,0,0]$ & $[v_{11},0,0]$ \\
 \hline
 2 & $\lambda$ & $1\longrightarrow 2$ & $[x_0,x_1,0]\longrightarrow [x_0,x_1,0]$ & $[v_{10},v_{11},0]$\\
 \hline
 3 & $\frac{\mu}{2}$ & $2\longrightarrow 1$ & $[x_0,x_1,x_2]\longrightarrow [x_1,x_2,0]$ & $[v_{21},v_{22},0]$\\
 \hline
 4 & $\frac{\mu}{2}$ & $2\longrightarrow 1$ & $[x_0,x_1,x_2]\longrightarrow [x_2,x_2,0]$ & $[v_{22},v_{22},0]$\\  
 \hline
  5 & $\lambda$ & $2\longrightarrow 2$ & $[x_0,x_1,x_2]\longrightarrow [x_0,x_1,0]$ & $[v_{20},v_{21},0]$\\ [1ex] 
 \hline
\end{tabular}
\end{center}
\caption{The SHS table under consideration in the proof of Proposition~\ref{prop:mm12-star-ps}}
\label{tab:mm12-sh-star}
\end{table}

\section{Proof of Proposition~\ref{prop:mm12-star-comparison}}
\label{app:prop:mm12-star-comparison}

As in the proof of Proposition~\ref{prop:mm12-comparison}, we first show that the ratio under study is monotonically increasing
with $\lambda$. First, we note that the ratio $\frac{\Delta_{M/M/1/2^*-FGFS}}
{\Delta_{M/M/1/2^*-PS}}
$
can be written as follows
\begin{equation}
\frac{4\lambda^5+14\lambda^4\mu+16\lambda^3\mu^2+14\lambda^2\mu^3+8\lambda\mu^4+2\mu^5}
{3\lambda^5+11\lambda^4\mu+15\lambda^3\mu^2+14\lambda^2\mu^3+8\lambda\mu^4+2\mu^5}.
\label{eq:app:prop:mm12-star-comparison-aux1}
\end{equation}

\begin{lemma}
The ratio $$
\frac{\Delta_{M/M/1/2^*-FGFS}}
{\Delta_{M/M/1/2^*-PS}}
$$
is an increasing function of $\lambda$.
\end{lemma}

\begin{proof}
The derivative of \eqref{eq:app:prop:mm12-star-comparison-aux1} with respect to $\lambda$ is
\begin{multline*}
\frac{20\lambda^4+56\lambda^3\mu+48\lambda^2\mu^2+28\lambda\mu^3+8\mu^4}
{3\lambda^5+11\lambda^4\mu+15\lambda^3\mu^2+14\lambda^2\mu^3+8\lambda\mu^4+2\mu^5}-\\
(15\lambda^4+44\lambda^3\mu+45\lambda^2\mu^2+28\lambda\mu^3+8\mu^4)\\\frac{4\lambda^5+14\lambda^4\mu+16\lambda^3\mu^2+14\lambda^2\mu^3+8\lambda\mu^4+2\mu^5}
{(3\lambda^5+11\lambda^4\mu+15\lambda^3\mu^2+14\lambda^2\mu^3+8\lambda\mu^4+2\mu^5)^2}.
\end{multline*}
We assume that the above expression is negative and we will see that it is an absurd. Thus, the derivative of 
\eqref{eq:app:prop:mm12-star-comparison-aux1} with respect to $\lambda$ is negative if and only if
\begin{multline*}
(20\lambda^4+56\lambda^3\mu+48\lambda^2\mu^2+28\lambda\mu^3+8\mu^4)\\
(3\lambda^5+11\lambda^4\mu+15\lambda^3\mu^2+14\lambda^2\mu^3+8\lambda\mu^4+2\mu^5)
<\\
(15\lambda^4+44\lambda^3\mu+45\lambda^2\mu^2+28\lambda\mu^3+8\mu^4)\\
(4\lambda^5+14\lambda^4\mu+16\lambda^3\mu^2+14\lambda^2\mu^3+8\lambda\mu^4+2\mu^5),
\end{multline*}
which expanding the polynomials it results
\begin{align*}
&60\lambda^9+388  \lambda^8\mu+1060\lambda^7\mu^2+1732\lambda^6\mu^3+1996\lambda^5\mu^4\\&+
1668\lambda^4\mu^5+1008\lambda^3\mu^6+432\lambda^2\mu^7+120\lambda\mu^8+16\mu^9 \\&<
60\lambda^9+386  \lambda^8\mu+1036\lambda^7\mu^2+1656\lambda^6\mu^3+1880\lambda^5\mu^4\\&+
1572\lambda^4\mu^5+968\lambda^3\mu^6+426\lambda^2\mu^7+120\lambda\mu^8+16\mu^9.
\end{align*}
We now simplify this expression and we obtain
\begin{multline*}
2  \lambda^8\mu+24\lambda^7\mu^2+76\lambda^6\mu^3+116\lambda^5\mu^4+
96\lambda^4\mu^5+40\lambda^3\mu^6\\+6\lambda^2\mu^7<0,
 \end{multline*}
which is clearly false since $\lambda$ and $\mu$ are positive. Therefore, the desired result follows.

\end{proof}

We now prove Proposition~\ref{prop:mm12-star-comparison} by studying the limit of the ratio
$\frac{\Delta_{M/M/1/2^*-FGFS}}
{\Delta_{M/M/1/2^*-PS}}$ when $\lambda$ tends to zero and to infinity. For the later limit, we get one, whereas for the former, we get 4/3. And the desired result follows. 

\section{Proof of Proposition~\ref{prop:mm12-mm12-star-comparison}}
\label{app:prop:mm12-mm12-star-comparison}

We show that the ratio $\frac{\Delta_{M/M/1/2-PS}}
{\Delta_{M/M/1/2^*-PS}}$ is increasing with $ \lambda$. We first provide the expression of the ratio under analysis:
\begin{multline*}
\frac{\Delta_{M/M/1/2-PS}}
{\Delta_{M/M/1/2^*-PS}}=\\
\frac{5\lambda^5+14\lambda^4\mu+17\lambda^3\mu^2+14\lambda^2\mu^3+8\lambda\mu^4+2\mu^5}{3\lambda^5+11\lambda^4\mu+15\lambda^3\mu^2+14\lambda^2\mu^3+8\lambda\mu^4+2\mu^5}.
\end{multline*}

We observe that the limit when $ \lambda\to 0$ (resp. when $\lambda\to \infty$) of the above expression 
is one (resp. is 5/3). Therefore, the proof ends by showing that the ratio 
$\frac{\Delta_{M/M/1/2^*-FGFS}}{\Delta_{M/M/1/2^*-PS}}$ is increasing with $ \lambda$. 

\begin{lemma}
$
\frac{\Delta_{M/M/1/2^-PS}}
{\Delta_{M/M/1/2^*-PS}}
$
is an increasing function of $\lambda$.
\label{lem:mm12-star-comparison-increasing-lambda}
\end{lemma}

\begin{proof}
The derivative of $\frac{\Delta_{M/M/1/2-PS}}
{\Delta_{M/M/1/2^*-PS}}$ with respect to $\lambda$ is
\begin{multline*}
\frac{25\lambda^4+56\lambda^3\mu+51\lambda^2\mu^2+28\lambda^2\mu^3+8\mu^4}{3\lambda^5+11\lambda^4\mu+15\lambda^3\mu^2+14\lambda^2\mu^3+8\lambda\mu^4+2\mu^5}\\
-(15\lambda^4+44\lambda^3\mu+45\lambda^2\mu^2+28\lambda\mu^3+8\mu^4)\\
\frac{5\lambda^5+14\lambda^4\mu+17\lambda^3\mu^2+14\lambda^2\mu^3+8\lambda\mu^4+2\mu^5}{(3\lambda^5+11\lambda^4\mu+15\lambda^3\mu^2+14\lambda^2\mu^3+8\lambda\mu^4+2\mu^5)^2}. 
\end{multline*}
This expression is positive if and only if
\begin{multline*}
(25\lambda^4+56\lambda^3\mu+51\lambda^2\mu^2+28\lambda^2\mu^3+8\mu^4)
\\(3\lambda^5+11\lambda^4\mu+15\lambda^3\mu^2+14\lambda^2\mu^3+8\lambda\mu^4+2\mu^5)\\
>(15\lambda^4+44\lambda^3\mu+45\lambda^2\mu^2+28\lambda\mu^3+8\mu^4)
\\(5\lambda^5+14\lambda^4\mu+17\lambda^3\mu^2+14\lambda^2\mu^3+8\lambda\mu^4+2\mu^5). 
\end{multline*}
Expanding the polynomials and simplifying, we get the following expression:
\begin{multline*}
13\lambda^8 \mu+48\lambda^7\mu^2+107\lambda^6\mu^3+148\lambda^5\mu^4
+120 \lambda^4\mu^5+56 \lambda^3\mu^6\\+12\lambda^2\mu^7>0,
\end{multline*}
which is clearly positive since $\lambda$ and $\mu$ are positive. Thus, the desired result follows. 
\end{proof}

%
%
\section{Proof of Proposition~\ref{prop:mm12-star2-ps}}
\label{app:prop:mm12-star2-ps}
The proof is very similar to that of Proposition~\ref{prop:mm12-star-ps}. We formulate the same SHS approach with the exception of transition $l=5$, whis is described in Table~\ref{tab:mm12-sh-star2}.

We again apply (35a) of \cite{yates2018age} and we get the following system of equations:

\begin{align*}
[v_{00},v_{01},v_{02}]\lambda=&b_0\pi_0 + \mu [v_{11},0,0]\\
[v_{10},v_{11},v_{12}](\lambda + \mu)=&b_1\pi_1 + \lambda [v_{00},0,0]+\frac{\mu}{2} [v_{21},v_{22},0]\\&+\frac{\mu}{2}[v_{22},v_{22},0]\\
[v_{20},v_{21},v_{22}](\lambda + \mu)=&b_2\pi_2 + \lambda [v_{10},v_{11},0]+\lambda [v_{20},v_{22},0].\\
\end{align*}

\begin{table}[t!]
\begin{center}
\begin{tabular}{|c c c c |c|} 
 \hline
 $l$ & $\lambda^{(l)}$ & $q\longrightarrow q'$ & $x\longrightarrow x'=A_lx$ & $v_{ql}A_l$ \\ \hline
 0 & $\lambda$ & $0\longrightarrow 1$ & $[x_0,0,0]\longrightarrow [x_0,0,0]$  & $[v_{00},0,0]$\\ 
 \hline
 1 & $\mu$ & $1\longrightarrow 0$ & $[x_0,x_1,0]\longrightarrow [x_1,0,0]$ & $[v_{11},0,0]$ \\
 \hline
 2 & $\lambda$ & $1\longrightarrow 2$ & $[x_0,x_1,0]\longrightarrow [x_0,x_1,0]$ & $[v_{10},v_{11},0]$\\
 \hline
 3 & $\frac{\mu}{2}$ & $2\longrightarrow 1$ & $[x_0,x_1,x_2]\longrightarrow [x_1,x_2,0]$ & $[v_{21},v_{22},0]$\\
 \hline
 4 & $\frac{\mu}{2}$ & $2\longrightarrow 1$ & $[x_0,x_1,x_2]\longrightarrow [x_2,x_2,0]$ & $[v_{22},v_{22},0]$\\  
 \hline
  5 & $\lambda$ & $2\longrightarrow 2$ & $[x_0,x_1,x_2]\longrightarrow [x_0,x_2,0]$ & $[v_{20},v_{22},0]$\\ [1ex] 
 \hline
\end{tabular}
\end{center}
\caption{The SHS table under consideration in the proof of Proposition~\ref{prop:mm12-star2-ps}}
\label{tab:mm12-sh-star2}
\end{table}

The solution of this system of linear equations is
\begin{align*} 
v_{00}&=\frac{3\lambda^2\mu^2+3\lambda\mu^3+\mu^4}{\lambda(\lambda+\mu)^2(\lambda^2+\mu^2+\lambda\mu)}\\
v_{10}&=\frac{5\lambda^4\mu+19\lambda^3\mu^2+21\lambda^2\mu^3+10\lambda\mu^4+2\mu^5}{2(\lambda+\mu)^4(\lambda^2+\mu^2+\lambda\mu)}\\
v_{11}&=\frac{2\lambda^2\mu+\lambda\mu^2}{(\lambda+\mu)^2(\lambda^2+\mu^2+\lambda\mu)}\\
v_{20}&=\frac{2\lambda^6+13\lambda^5\mu+31\lambda^4\mu^2+29 \lambda^3\mu^3+12\lambda^2\mu^4+2\lambda\mu^5}
{2\mu(\lambda+\mu)^4(\lambda^2+\mu^2+\lambda\mu)}\\
v_{21}&=\frac{2\lambda^4+5\lambda^3\mu+2\lambda^2\mu^2}{(\lambda+\mu)^3(\lambda^2+\mu^2+\lambda\mu)}\\
v_{22}&=\frac{\lambda^2}{(\lambda+\mu)(\lambda^2+\mu^2+\lambda\mu)}.
 \end{align*}

According to Theorem 4 of  \cite{yates2018age}, the desired value is obtained by summing $v_{00}$, $v_{10}$ and $v_{20}$. 
And the desired result follows. 

\section{Proof of Proposition~\ref{prop:mm12-star2-fgfs}}
\label{app:prop:mm12-star2-fgfs}
The proof is very similar to that of Proposition~\ref{prop:mm12-star-ps} and Proposition~\ref{prop:mm12-star2-ps}. We formulate the same SHS approach, but now as we are operating unde FGFS discipline some transitions will be different, whis are described in Table~\ref{tab:mm12-fgfs-star2}.

We apply (35a) of \cite{yates2018age} and we get the following system of equations:

\begin{align*}
[v_{00},v_{01},v_{02}]\lambda=&b_0\pi_0 + \mu [v_{11},0,0]\\
[v_{10},v_{11},v_{12}](\lambda + \mu)=&b_1\pi_1 + \lambda [v_{00},0,0]+\mu[v_{21},v_{22},0]\\
[v_{20},v_{21},v_{22}](\lambda + \mu)=&b_2\pi_2 + \lambda [v_{10},v_{11},0]+\lambda [v_{20},v_{22},0].\\
\end{align*}

\begin{table}[t!]
\begin{center}
\begin{tabular}{|c c c c |c|} 
 \hline
 $l$ & $\lambda^{(l)}$ & $q\longrightarrow q'$ & $x\longrightarrow x'=A_lx$ & $v_{ql}A_l$ \\ \hline
 0 & $\lambda$ & $0\longrightarrow 1$ & $[x_0,0,0]\longrightarrow [x_0,0,0]$  & $[v_{00},0,0]$\\ 
 \hline
 1 & $\mu$ & $1\longrightarrow 0$ & $[x_0,x_1,0]\longrightarrow [x_1,0,0]$ & $[v_{11},0,0]$ \\
 \hline
 2 & $\lambda$ & $1\longrightarrow 2$ & $[x_0,x_1,0]\longrightarrow [x_0,x_1,0]$ & $[v_{10},v_{11},0]$\\
 \hline
 3 & $\mu$ & $2\longrightarrow 1$ & $[x_0,x_1,x_2]\longrightarrow [x_1,x_2,0]$ & $[v_{21},v_{22},0]$\\
 \hline
 4 & $\lambda$ & $2\longrightarrow 2$ & $[x_0,x_1,x_2]\longrightarrow [x_0,x_2,0]$ & $[v_{20},v_{22},0]$\\  
 \hline
\end{tabular}
\end{center}
\caption{The SHS table under consideration in the proof of Proposition~\ref{prop:mm12-star2-fgfs}}
\label{tab:mm12-fgfs-star2}
\end{table}

The solution of this system of linear equations is
\begin{align*} 
v_{00}&=\frac{3\lambda^2\mu^2+3\lambda\mu^3+\mu^4}{\lambda(\lambda+\mu)^2(\lambda^2+\mu^2+\lambda\mu)}\\
v_{10}&=\frac{3\lambda^4\mu+11\lambda^3\mu^2+11\lambda^2\mu^3+5\lambda\mu^4+\mu^5}{(\lambda+\mu)^4(\lambda^2+\mu^2+\lambda\mu)}\\
v_{11}&=\frac{2\lambda^2\mu+\lambda\mu^2}{(\lambda+\mu)^2(\lambda^2+\mu^2+\lambda\mu)}\\
v_{20}&=\frac{\lambda^6+7\lambda^5\mu+17\lambda^4\mu^2+15\lambda^3\mu^3+6\lambda^2\mu^4+\lambda\mu^5}
{\mu(\lambda+\mu)^4(\lambda^2+\mu^2+\lambda\mu)}\\
v_{21}&=\frac{2\lambda^4+5\lambda^3\mu+2\lambda^2\mu^2}{(\lambda+\mu)^3(\lambda^2+\mu^2+\lambda\mu)}\\
v_{22}&=\frac{\lambda^2}{(\lambda+\mu)(\lambda^2+\mu^2+\lambda\mu)}.
 \end{align*}

According to Theorem 4 of  \cite{yates2018age}, the desired value is obtained by summing $v_{00}$, $v_{10}$ and $v_{20}$. 
And the desired result follows. 

\section{Proof of Proposition~\ref{prop:mm12-star2-comparison}}
\label{app:prop:mm12-star2-comparison}

We are first going to write the expression of the ratio.

\begin{multline*}
\frac{\Delta_{M/M/1/2^{**}-FGFS}}{\Delta_{M/M/1/2^{**}-PS}}=\\
\frac{2\lambda^6+12\lambda^5\mu+28\lambda^4\mu^2+30\lambda^3\mu^3+22\lambda^2\mu^4+10\lambda\mu^5+2\mu^6}{2\lambda^6+11\lambda^5\mu+25\lambda^4\mu^2+29\lambda^3\mu^3+22\lambda^2\mu^4+10\lambda\mu^5+2\mu^6}.
\label{eq:mm12-star2-comparison-fgfs-ps}
\end{multline*}

Since $\lambda>0$ and $\mu>0$ it is easily seen that $$\Delta_{M/M/1/2^{**}-PS}\leq \Delta_{M/M/1/2^{**}-FGFS}.$$  Therefore, we know that 
\begin{equation}
1\leq \frac{\Delta_{M/M/1/2^{**}-FGFS}}{\Delta_{M/M/1/2^{**}-PS}}.
\label{eq:mm12star2-fgfs-ps-aux}
\end{equation}

Since this ratio is not an increasing function of $\lambda$, we now want to find the maximum value of it, that way we will have proven Proposition~\ref{prop:mm12-star2-comparison}. In order to that, we present the following result.

\begin{lemma}
The ratio  $\frac{\Delta_{M/M/1/2^{**}-FGFS}}{\Delta_{M/M/1/2^{**}-PS}}$ takes its maximum value at $\rho=2.3943$ and it is 
\begin{equation}
\frac{\Delta_{M/M/1/2^{**}-FGFS}}{\Delta_{M/M/1/2^{**}-PS}}=1.0731
\end{equation}
\end{lemma}

\begin{proof}
First we rewrite the ratio by dividing $\mu^6$ in the numerator and the denominator and we get the following
\begin{multline*}
\frac{\Delta_{M/M/1/2^{**}-FGFS}}{\Delta_{M/M/1/2^{**}-PS}}=\\
\frac{2\rho^6+12\rho^5+28\rho^4+30\rho^3+22\rho^2+10\rho+2}{2\rho^6+11\rho^5+25\rho^4+29\rho^3+22\rho^2+10\rho+2}.
\end{multline*}
The derivative of $\frac{\Delta_{M/M/1/2^{**}-FGFS}}{\Delta_{M/M/1/2^{**}-PS}}$ with respect to $\rho$ is
\begin{multline*}
\frac{12\rho^5+60\rho^4+112\rho^3+90\rho^2+44\rho+10}{2\rho^6+11\rho^5+25\rho^4+29\rho^3+22\rho^2+10\rho+2}\\
-(12\rho^5+55\rho^4+100\rho^3+87\rho^2+44\rho+10)\\
\frac{2\rho^6+12\rho^5+28\rho^4+30\rho^3+22\rho^2+10\rho+2}{(2\rho^6+11\rho^5+25\rho^4+29\rho^3+22\rho^2+10\rho+2)^2}. 
\end{multline*}

We set the derivative equal to zero and we get the following result
\begin{multline*}
(12\rho^5+60\rho^4+112\rho^3+90\rho^2+44\rho+10)\\(2\rho^6+11\rho^5+25\rho^4+29\rho^3+22\rho^2+10\rho+2)
\\-(12\rho^5+55\rho^4+100\rho^3+87\rho^2+44\rho+10)\\(2\rho^6+12\rho^5+28\rho^4+30\rho^3+22\rho^2+10\rho+2)=0
\end{multline*}
Expanding that expression we get
\begin{multline}
-2\rho^{10}-12\rho^9-14\rho^8+36\rho^7+128\rho^6\\+172\rho^5+122\rho^4+44\rho^3+6\rho^2=0
\label{eq:aux-func}
\end{multline}

Since $\lambda>0$ and $\mu>0$ then $\rho$ must be positive, and the only postive root of that expression is $\rho=2.3943$. Therefore, this ratio is larger than one from \eqref{eq:mm12star2-fgfs-ps-aux} and it is equal to one when $\rho\to 0$ and $\rho\to \infty$. Therefore, it has a unique maximum when $\rho$ is positive, which is achieved for $\rho=2.3943$.
We evaluate $\rho=2.3943$ on our ratio and the desired result follows.
\end{proof}

\section{Proof of Proposition~\ref{prop:mm12-star1-2-ps-comparison}}
\label{app:prop:mm12-star1-2-ps-comparison}
We have 
\begin{multline}\frac{\Delta_{M/M/1/2^*-PS}}{\Delta_{M/M/1/2^{**}-PS}}=\\
\frac{3\lambda^6+14\lambda^5\mu+26\lambda^4\mu^2+29\lambda^3\mu^3+22\lambda^2\mu^4+10\lambda\mu^5+2\mu^6}{2\lambda^6+11\lambda^5\mu+25\lambda^4\mu^2+29\lambda^3\mu^3+22\lambda^2\mu^4+10\lambda\mu^5+2\mu^6}
\label{eq:mm12-comparison-star2-aux}
\end{multline}

Thus, taking into account that the rhs of \eqref{eq:mm12-comparison-star2-aux} tends to 1 when $\lambda\to 0$ and to $\frac{3}{2}$ when $\lambda\to \infty$, 
the desired result follows if we show that the rhs of \eqref{eq:mm12-comparison-star2-aux} is increasing with $\lambda$, which we proof in the following result.
\begin{lemma}
The rhs of \eqref{eq:mm12-comparison-star2-aux} is an increasing function of $\lambda$ for all $\mu>0$.
\label{lem:mm12-comparison-increasing-lambda}
\end{lemma}

 \begin{proof}
We compute the derivative of the ratio with respect to $\lambda$ and it results
   
    \begin{multline*}\frac{18\lambda^5+70\lambda^4\mu+104\lambda^3\mu^2+87\lambda^2\mu^3+44\lambda\mu^4+10\mu^5}{2\lambda^6+11\lambda^5\mu+25\lambda^4\mu^2+29\lambda^3\mu^3+22\lambda^2\mu^4+10\lambda\mu^5+2\mu^6}\\-(12\lambda^5+55\lambda^4\mu+100\lambda^3\mu^2+87\lambda^2\mu^3+44\lambda\mu^4+10\mu^5)\\\frac{(3\lambda^6+14\lambda^5\mu+26\lambda^4\mu^2+29\lambda^3\mu^3+22\lambda^2\mu^4+10\lambda\mu^5+2\mu^6)}{(2\lambda^6+11\lambda^5\mu+25\lambda^4\mu^2+29\lambda^3\mu^3+22\lambda^2\mu^4+10\lambda\mu^5+2\mu^6)^2}.
\end{multline*}
    The above expression is positive if and only if
    \begin{multline*}(18\lambda^5+70\lambda^4\mu+104\lambda^3\mu^2+87\lambda^2\mu^3+44\lambda\mu^4+10\mu^5)\\(2\lambda^6+11\lambda^5\mu+25\lambda^4\mu^2+29\lambda^3\mu^3+22\lambda^2\mu^4+10\lambda\mu^5+2\mu^6)\\>(12\lambda^5+55\lambda^4\mu+100\lambda^3\mu^2+87\lambda^2\mu^3+44\lambda\mu^4+10\mu^5)\\(3\lambda^6+14\lambda^5\mu+26\lambda^4\mu^2+29\lambda^3\mu^3+22\lambda^2\mu^4+10\lambda\mu^5+2\mu^6),\end{multline*}
which simplifying we obtain that
\begin{multline*}
5\lambda^{10}\mu+46\lambda^9\mu^2+151\lambda^8\mu^3+262\lambda^7\mu^4+277\lambda^6\mu^5\\+176\lambda^5\mu^6+60\lambda^4\mu^7+8\lambda^3\mu^8> 0,\end{multline*}
which is true for all $\mu>0$. And the desired result follows.
\end{proof}

\section{Proof of Proposition~\ref{prop:mm11-mm12-star-ps-comparison}}
\label{app:mm11-mm12-star-ps-comparison}
We have 
\begin{multline}
\frac{\Delta_{M/M/1/1}}{\Delta_{M/M/1/2^{*}-PS}}=\\
\frac{4\lambda^5+12\lambda^4\mu+18\lambda^3\mu^2+16\lambda^2\mu^3+8\lambda\mu^4+2\mu^5}{3\lambda^5+11\lambda^4\mu+15\lambda^3\mu^2+14\lambda^2\mu^3+8\lambda\mu^4+2\mu^5}
\label{eq:mm11-mm12-star-ps-comparison}
\end{multline}

Thus, taking into account that the rhs of \eqref{eq:mm11-mm12-star-ps-comparison} tends to 1 when $\lambda\to 0$ and to $\frac{4}{3}$ when $\lambda\to \infty$, 
the desired result follows if we show that the rhs of \eqref{eq:mm11-mm12-star-ps-comparison} is increasing with $\lambda$, which we proof in the following result.
\begin{lemma}
The rhs of \eqref{eq:mm11-mm12-star-ps-comparison} is an increasing function of $\lambda$ for all $\mu>0$.
\label{lem:mm12-comparison-increasing-lambda}
\end{lemma}

 \begin{proof}
We compute the derivative of the ratio with respect to $\lambda$ and it results
   
    \begin{multline*}\frac{20\lambda^4+48\lambda^3\mu+54\lambda^2\mu^2+32\lambda\mu^3+8\mu^4}{3\lambda^5+11\lambda^4\mu+15\lambda^3\mu^2+14\lambda^2\mu^3+8\lambda\mu^4+2\mu^5}\\-(15\lambda^4+44\lambda^3\mu+45\lambda^2\mu^2+28\lambda\mu^3+8\mu^4)\\\frac{(4\lambda^5+12\lambda^4\mu+18\lambda^3\mu^2+16\lambda^2\mu^3+8\lambda\mu^4+2\mu^5)}{(3\lambda^5+11\lambda^4\mu+15\lambda^3\mu^2+14\lambda^2\mu^3+8\lambda\mu^4+2\mu^5)^2}.
\end{multline*}
    The above expression is positive if and only if
    \begin{multline*}(20\lambda^4+48\lambda^3\mu+54\lambda^2\mu^2+32\lambda\mu^3+8\mu^4)\\(3\lambda^5+11\lambda^4\mu+15\lambda^3\mu^2+14\lambda^2\mu^3+8\lambda\mu^4+2\mu^5)\\>(15\lambda^4+44\lambda^3\mu+45\lambda^2\mu^2+28\lambda\mu^3+8\mu^4)\\(4\lambda^5+12\lambda^4\mu+18\lambda^3\mu^2+16\lambda^2\mu^3+8\lambda\mu^4+2\mu^5),\end{multline*}
which simplifying we obtain that
\begin{multline*}
8\lambda^{8}\mu+12\lambda^7\mu^2+6\lambda^6\mu^3+16\lambda^5\mu^4+46\lambda^4\mu^5\\+56\lambda^3\mu^6+34\lambda^2\mu^7+8\lambda\mu^8> 0,\end{multline*}
which is true for all $\mu>0$. And the desired result follows.
\end{proof}

\section{Proof of Proposition~\ref{prop:mm11-mm12-star2-ps-comparison2}}
\label{app:prop:mm11-mm12-star2-ps-comparison2}

We have the following ratio, 
\begin{multline}
\frac{\Delta_{M/M/1/2-PS}}{\Delta_{M/M/1/1}}=\\
\frac{5\lambda^4+9\lambda^3\mu+8\lambda^2\mu^2+6\lambda^2\mu^3+2\mu^4}{4\lambda^4+8\lambda^3\mu+10\lambda^2\mu^2+6\lambda^2\mu^3+2\mu^4}
\label{eq:mm11-mm12-ps-comparison}
\end{multline}
When $\lambda\in (0,\mu]$, then we have that $\Delta_{M/M/1/1}\geq\Delta_{M/M/1/2-PS}$ and when $\lambda\in [\mu,\infty)$, we have $\Delta_{M/M/1/1}\leq\Delta_{M/M/1/2-PS}$. So we will study each case separately. 

In the case of $\lambda\in (0,\mu]$, we want to find the minimum value of the ratio.

\begin{lemma} If $\lambda\in (0,\mu]$, then we have the following
\begin{equation}
0.9641\leq\frac{\Delta_{M/M/1/2-PS}}{\Delta_{M/M/1/1}}\leq 1
\label{eq:mm11-mm12-ps-aux}
\end{equation}
\end{lemma}
\begin{proof}
First we rewrite the ratio by dividing $\mu^4$ in the numerator and the denominator. So we get the following
\begin{equation}
\frac{\Delta_{M/M/1/2-PS}}{\Delta_{M/M/1/1}}=
\frac{5\rho^4+9\rho^3+8\rho^2+6\rho^2+2}{4\rho^4+8\rho^3+10\rho^2+6\rho^2+2}.
\label{eq:ratio-mm12-ps-mm11}
\end{equation}
The derivative of the ratio with respect to $\rho$ is
\begin{multline*}
\frac{20\rho^3+27\rho^2+16\rho+6}{4\rho^4+8\rho^3+10\rho^2+6\rho^2+2}\\
-(16\rho^3+24\rho^2+20\rho+6)
\frac{5\rho^4+9\rho^3+8\rho^2+6\rho^2+2}{(4\rho^4+8\rho^3+10\rho^2+6\rho^2+2)^2}. 
\end{multline*}

We set the derivative equal to zero and we get the following result
\begin{multline*}
(20\rho^3+27\rho^2+16\rho+6)(4\rho^4+8\rho^3+10\rho^2+6\rho^2+2)
\\-(16\rho^3+24\rho^2+20\rho+6)(5\rho^4+9\rho^3+8\rho^2+6\rho^2+2)=0
\end{multline*}
Expanding that expression we get
\begin{multline}
4\rho^6+36\rho^5+44\rho^4+20\rho^3-6\rho^2-8\rho=0
\label{eq:aux2-func}
\end{multline}

Since $\lambda>0$ and $\mu>0$ then $\rho$ must be positive, and the only postive root of that expression is $\rho=0.4697$.  Besides, 
$$
\frac{5\rho^4+9\rho^3+8\rho^2+6\rho^2+2}{4\rho^4+8\rho^3+10\rho^2+6\rho^2+2}
$$
is clearly smaller or equal than one and equal to one when $\rho\to 0$ and $\rho\to\infty$. Therefore, since we have shown that it has a single local critical point, it is a minimum. 
%
We evaluate $\rho=0.4697$ on our ratio and the desired result follows.
\end{proof}

We now focus in the case of $\lambda\in[\mu,\infty)$. In this case, we have $\rho=\frac{\lambda}{\mu}$ where $\lambda>\mu$. Knowing that, it is clearly visible that the expression \eqref{eq:aux2-func} is always positive.  So \eqref{eq:aux2-func}  is an increasing function on $\lambda$. Then as the ratio tends to $\frac{5}{4}$ when $\lambda\to\infty$, we have the following result.
\begin{equation}
1\leq \frac{\Delta_{M/M/1/2-PS}}{\Delta_{M/M/1/1}}\leq \frac{5}{4}.
\label{eq:mm12-ps-mm11-ineq}
\end{equation}
 Taking into account \eqref{eq:mm11-mm12-ps-aux} and \eqref{eq:mm12-ps-mm11-ineq} the desired result is achieved.

\section{Proof of Proposition~\ref{prop:mm11-star-mm12-star2-ps-comparison}}
\label{app:prop:mm11-star-mm12-star2-ps-comparison}

We first going to write the expression of the ratio.

\begin{multline*}
\frac{\Delta_{M/M/1/2^{*}-PS}}{\Delta_{M/M/1/1^*}}=\\
\frac{2\lambda^6+11\lambda^5\mu+25\lambda^4\mu^2+29\lambda^3\mu^3+22\lambda^2\mu^4+10\lambda\mu^5+2\mu^6}{2\lambda^6+10\lambda^5\mu+22\lambda^4\mu^2+28\lambda^3\mu^3+22\lambda^2\mu^4+10\lambda\mu^5+2\mu^6}.
\label{eq:mm11-star-mm12-star}
\end{multline*}

Since $\lambda>$0 and $\mu>$0 it is easily seen that $\Delta_{M/M/1/1^*}\leq \Delta_{M/M/1/2^*-PS}$.  So we know that 
\begin{equation}
1\leq \frac{\Delta_{M/M/1/2^{*}-PS}}{\Delta_{M/M/1/1^*}}.
\end{equation}

Since this ratio is not an increasing function of $\lambda$, we now want to find the maximum value of it, that way we will have proven Proposition~\ref{prop:mm11-star-mm12-star2-ps-comparison}. In order to that, we present the following result.

\begin{lemma}
The ratio  $\frac{\Delta_{M/M/1/2^{*}-PS}}{\Delta_{M/M/1/1^*}}$ takes its maximum value at $\rho=2.3943$ and it is 
\begin{equation}
\frac{\Delta_{M/M/1/2^{*}-PS}}{\Delta_{M/M/1/1^*}}=1.0788
\end{equation}
\end{lemma}

\begin{proof}
First we rewrite the ratio by dividing $\mu^6$ in the numerator and the denominator. So we get the following
\begin{multline*}
\frac{\Delta_{M/M/1/2^{*}-PS}}{\Delta_{M/M/1/1^*}}=\\
\frac{2\rho^6+11\rho^5+25\rho^4+29\rho^3+22\rho^2+10\rho+2}{2\rho^6+10\rho^5+22\rho^4+28\rho^3+22\rho^2+10\rho+2}.
\end{multline*}
The derivative of $\frac{\Delta_{M/M/1/2^{*}-PS}}{\Delta_{M/M/1/1^*}}$ with respect to $\rho$ is
\begin{multline*}
\frac{12\rho^5+55\rho^4+100\rho^3+87\rho^2+44\rho+10}{2\rho^6+10\rho^5+22\rho^4+28\rho^3+22\rho^2+10\rho+2}\\
-(12\rho^5+50\rho^4+88\rho^3+84\rho^2+44\rho+10)\\
\frac{2\rho^6+11\rho^5+25\rho^4+29\rho^3+22\rho^2+10\rho+2}{(2\rho^6+10\rho^5+22\rho^4+28\rho^3+22\rho^2+10\rho+2)^2}. 
\end{multline*}

We set the derivative equal to zero and we get the following result
\begin{multline*}
(12\rho^5+55\rho^4+100\rho^3+87\rho^2+44\rho+10)\\(2\rho^6+10\rho^5+22\rho^4+28\rho^3+22\rho^2+10\rho+2)
\\-(12\rho^5+50\rho^4+88\rho^3+84\rho^2+44\rho+10)\\(2\rho^6+11\rho^5+25\rho^4+29\rho^3+22\rho^2+10\rho+2)=0
\end{multline*}
Expanding that expression we get
\begin{multline}
-2\rho^{10}-12\rho^9-14\rho^8+36\rho^7+128\rho^6\\+172\rho^5+122\rho^4+44\rho^3+6\rho^2=0
\label{eq:aux-func2}
\end{multline}

Since $\lambda>0$ and $\mu>0$ then $\rho$ must be positive, and the only positive root of that expression is $\rho=2.3943$.  Besides, we have that 
$$
\frac{2\rho^6+11\rho^5+25\rho^4+29\rho^3+22\rho^2+10\rho+2}{2\rho^6+10\rho^5+22\rho^4+28\rho^3+22\rho^2+10\rho+2},
$$
is clearly larger than one when $\rho\in (0,\infty)$ and tends to one when $\rho\to 0$ and $\rho\to \infty$. Therefore, this ratio has a unique maximum when $\rho\in (0,\infty)$.  
We evaluate $\rho=2.3943$ on our ratio and the desired result follows.
\end{proof}
%
%

\section{Proof of Lemma~\ref{lem:mm1-lowerbound}. }
\label{app:lem:mm1-lowerbound}.

We also model the system using the SHS methodology. In this case, the Markov chain we consider is $\mathcal Q= \{0,1,2,\dots\}$, 
which is a birth-death process with birth rate $\lambda$ and death rate $\mu$. For this Markov chain, the stationary distribution is clearly $\pi_i=(1-\rho)\rho^i.$ For the continuous state, we will only focus on the transitions of state zero. Indeed, the idea of the proof is to apply Theorem 4 of  \cite{yates2018age} as follows
$$
\Delta_{M/M/1-PS}=\sum_{q\in\mathcal Q}v_{q0}>v_{00}.
$$

Thus, in the SHS table under consideration, we only show the transitions related to state zero as well as the values of the continuous state of state zero. This is represented in Table~\ref{tab:mm1}.

\begin{table}[t!]
\begin{center}
\begin{tabular}{|c c c c |c|} 
 \hline
 $l$ & $\lambda^{(l)}$ & $q\longrightarrow q'$ & $x\longrightarrow x'=A_lx$ & $v_{ql}A_l$ \\ \hline
 0 & $\lambda$ & $0\longrightarrow 1$ & $[x_0,...]\longrightarrow [x_0,...]$  & $[v_{00},...]$\\ 
 \hline
 1 & $\mu$ & $1\longrightarrow 0$ & $[x_0,...]\longrightarrow [x_1,...]$ & $[v_{11},...]$ \\\hline
\end{tabular}
\end{center}
\caption{The SHS table under consideration in the proof of Lemma~\ref{lem:mm1-lowerbound}}
\label{tab:mm1}
\end{table}

We omit the explanation of the transitions represented in Table~\ref{tab:mm1} because they coincide with the transitions 0 and 1 of the SHS table of Proposition~\ref{prop:mm12-ps}.
Now, we apply (35a) of \cite{yates2018age} to the SHS of Table~\ref{tab:mm1} and we get
$$
\lambda v_{00}=\pi_0+v_{11}\mu>\pi_0 \iff  v_{00}=\frac{\mu-\lambda}{\lambda\mu}.
$$ 
And the desired result follows.

\section{Proof of Proposition~\ref{prop:mm1-ratio}. }
\label{app:prop:mm1-ratio}.

Under Conjecture~\ref{conject:mm1ps}, we know that $C(\rho)\leq \frac{\rho^2}{1-\rho}$ and $C(\rho)\geq \frac{0.75\rho}{(1-\rho)^{\frac{1}{2}}}$, therefore
$$
\Delta_{M/M/1-FGFS}\leq \Delta_{M/M/1-PS}\leq \frac{1}{\mu}\left(1+\frac{1}{\rho}+\frac{0.75\rho}{(1-\rho)^{\frac{1}{2}}}\right).$$

Therefore, the proof ends if we show that
$
\frac{\Delta_{M/M/1-FGFS}}{\Delta_{M/M/1-PS}} 
$
is unbounded from above. That is,
\begin{align*}
\frac{1+\frac{1}{\rho}+\frac{\rho^2}{1-\rho}}
{1+\frac{1}{\rho}+C(\rho)}\geq 
\frac{1+\frac{1}{\rho}+\frac{\rho^2}{1-\rho}}
{1+\frac{1}{\rho}+\frac{0.75\rho}{(1-\rho)^{\frac{1}{2}}}}=
1+\frac{\frac{\rho^2}{1-\rho}-\frac{0.75\rho}{(1-\rho)^{\frac{1}{2}}}}{1+\frac{1}{\rho}+\frac{0.75\rho}{(1-\rho)^{\frac{1}{2}}}}
\end{align*}
and the last expression tends to infinity when $\rho\to 1$ since 
$$
\frac{\frac{\rho^2}{1-\rho}-\frac{0.75\rho}{(1-\rho)^{\frac{1}{2}}}}{1+\frac{1}{\rho}+\frac{0.75\rho}{(1-\rho)^{\frac{1}{2}}}}=
\frac{\frac{\rho^2}{(1-\rho)^{\frac{1}{2}}}-0.75\rho}{\left(1+\frac{1}{\rho}\right)(1-\rho)^{\frac{1}{2}}+0.75\rho},
$$
and the numerator tends to infinity and the denominator to 0.75 when $\rho\to 1$.

%
%

\ifCLASSOPTIONcaptionsoff
  \newpage
\fi



\bibliographystyle{IEEEtran}
\bibliography{IEEEabrv,aoi-ps}

\begin{IEEEbiographynophoto}{Beñat Gandarias}
obtained the Degree in Mathematics from the University of the Basque Country, UPV/EHU, in 2023. He is currently a research staff in the Mathematics department of the UPV/EHU. His research interests cover the analysis of queueing disciplines from the Age of Information perspective. 
\end{IEEEbiographynophoto}

\begin{IEEEbiographynophoto}{Josu Doncel}
 obtained from the University of the Basque Country (UPV/EHU) the Industrial Engineering degree in 2007, the Mathematics degree in 2010 and, in 2011, the Master degree in Applied Mathematics and Statistics. I received in 2015 the PhD degree from Université de Toulouse (France). He is currently associate professor in the Mathematics department of the UPV/EHU. He has previously held research positions at LAAS-CNRS (France), INRIA Grenoble (France) and BCAM-Basque Center for Applied Mathematics (Spain), teaching positions at ENSIMAG (France), INSA-Toulouse (France) and IUT-Blagnac (France) and invited researcher positions at Laboratory of Signals and Systems of CentraleSupelec (France), at Inria Paris (France) and at Laboratory David (France).
His research interests are the modeling, optimization and performance evaluation of distributed stochastic systems such as telecommunications and energy networks.
\end{IEEEbiographynophoto}

\begin{IEEEbiographynophoto}{Mohamad Assaad}
received the M.Sc and
PhD degrees, both in telecommunications, from
Telecom ParisTech, Paris, France, in 2002 and
2006, respectively. Since 2006, he has been with
the Telecommunications Department at CentraleSupelec, where he is currently a professor. He
is also a researcher at the Laboratoire des Signaux et Systemes (L2S, CNRS) and holds the
5G Chair. He has co-authored 1 book and more
than 120 publications in journals and conference
proceedings and has served as TPC member or TPC co-chair for toptier international conferences, including TPC co-chair for IEEE WCNC 21, IEEE Globecom 20 Mobile and Wireless networks Symposium Cochair, etc. He is an Editor for the IEEE Wireless Communications Letters
and the Journal of Communications and Information Networks. He served
also as a guest co-editor for a special issue of the IEEE transactions on
Network Science and Engineering. He has given in the past successful
tutorials on several topics related to 5G systems, and Age of Information
at various conferences including IEEE ISWCS 15, IEEE WCNC 16, 
and IEEE ICC 21 conferences. His research interests include 5G and 
beyond systems, fundamental networking aspects of wireless systems, Age
of Information, resource optimization and Machine Learning in wireless
networks.
\end{IEEEbiographynophoto}

\end{document}